\documentclass[12pt]{article}
\usepackage{amsmath}
\usepackage{graphicx}
\usepackage[sort]{natbib}
\usepackage{url} % not crucial - just used below for the URL 
\usepackage[utf8]{inputenc}
\usepackage[english]{babel}
\usepackage{amsthm}

\usepackage{natbib}

\usepackage{hyperref}       % hyperlinks
\usepackage{url}            % simple URL typesetting
\usepackage{booktabs}       % professional-quality tables
\usepackage{amsfonts}       % blackboard math symbols
\usepackage{nicefrac}       % compact symbols for 1/2, etc.
\usepackage{microtype}      % microtypography
\usepackage{lipsum}		% Can be removed after putting your text content
\usepackage{graphicx}
\usepackage{natbib}
\usepackage{amsmath,amssymb}
\usepackage{doi}

\newtheorem{thm}{Theorem}

\newtheorem{assumption}{Assumption}

\newtheorem{lemma}{Lemma}

\usepackage{mathtools}
\DeclarePairedDelimiter\ceil{\lceil}{\rceil}

\DeclareMathOperator*{\argmin}{arg\,min}
\DeclareMathOperator*{\argmax}{arg\,max}

\usepackage{algorithm}
\usepackage{algorithmic}
\usepackage{setspace}
\usepackage{multirow} 
\usepackage{booktabs}
\usepackage[toc,page]{appendix}
\usepackage[english]{babel}
\usepackage{enumerate}

\newcommand{\blind}{0}

\begin{document}

\def\spacingset#1{\renewcommand{\baselinestretch}%
{#1}\small\normalsize} \spacingset{1}

%%%%%%%%%%%%%%%%%%%%%%%%%%%%%%%%%%%%%%%%%%%%%%%%%%%%%%%%%%%%%%%%%%%%%%%%%%%%%%

\if0\blind
{
  \title{\bf  Change Point Detection in Pairwise Comparison Data with Covariates}
  \author{Yi Han\footnote{yyihan@ucdavis.edu
    }\hspace{.2cm}\\
    Department of Statistics, University of California, Davis\\
    and \\
    Thomas C. M. Lee\thanks{Corresponding Author. tcmlee@ucdavis.edu.  The authors gratefully acknowledge the support of the National Science Foundation under grants DMS-2113605 and DMS-2210388.} \\
    Department of Statistics, University of California, Davis}
  \maketitle
} \fi

\if1\blind
{
  \bigskip
  \bigskip
  \bigskip
  \begin{center}
    {\LARGE\bf Title}
\end{center}
  \medskip
} \fi

\bigskip
\begin{abstract}
This paper introduces the novel piecewise stationary covariate-assisted ranking estimation (PS-CARE) model for analyzing time-evolving pairwise comparison data, enhancing item ranking accuracy through the integration of covariate information. By partitioning the data into distinct, stationary segments, the PS-CARE model adeptly detects temporal shifts in item rankings, known as change points, whose number and positions are initially unknown. Leveraging the minimum description length (MDL) principle, this paper establishes a statistically consistent model selection criterion to estimate these unknowns. The practical optimization of this MDL criterion is done with the pruned exact linear time (PELT) algorithm. Empirical evaluations reveal the method's promising performance in accurately locating change points across various simulated scenarios. An application to an NBA dataset yielded meaningful insights that aligned with significant historical events, highlighting the method's practical utility and the MDL criterion's effectiveness in capturing temporal ranking changes.  To the best of the authors' knowledge, this research pioneers change point detection in pairwise comparison data with covariate information, representing a significant leap forward in the field of dynamic ranking analysis.
\end{abstract}

\noindent%
\vfill

\newpage
\spacingset{1.75} % DON'T change the spacing!

\section{Introduction}
\label{sec:intro}
%add a few sentences about ranking?

The ranking problem has long held a pivotal role in numerous real-life applications, spanning diverse areas such as recommendation systems \citep{baltrunas2010group}, university admissions \citep{strobl2011accounting}, 
sports analytics \citep{mchale2011bradley, cattelan2013dynamic}, election candidate evaluations \citep{plackett1975analysis}, and web search algorithms \citep{dwork2001rank}. These rankings not only offer insights into the comparative quality of entities but also shape subsequent decisions, highlighting the problem's significance. Over the years, this has led to a collection of methods designed to tackle ranking problems, including 
%\cite{hunter2004mm}, 
\cite{huang2006ranking},
%\cite{cossock2008statistical}, 
\cite{buyse2010generalized}, 
%\cite{weng2011bayesian},
%\cite{shah2018simple},
%\cite{fan2022uncertainty},
%\cite{pearce2022unified},
and \cite{newman2023efficient}.
%and \cite{osting2014optimal}.

Of all these, the Bradley-Terry-Luce (BTL) model, introduced by \cite{bradley1952rank} and later expanded upon by \cite{luce1959individual}, stands out due to its widespread adoption. This model postulates an intrinsic score for each item being compared. Given $n$ items undergoing pairwise comparison, denoted by intrinsic scores $\{\theta_i\}_{i=1}^{n}$, the probability that item $i$ ranks above item $j$ is given by:
\begin{equation}
P(\text{item $i$ beats item $j$}) = \frac{e^{\theta_i}}{e^{\theta_i}+e^{\theta_j}}.
\label{eqn:btl}
\end{equation}

The comparison result is denoted as $y\in \{0,1\}$. Several algorithms are also proposed to solve the problem \cite{caron2012efficient}. However, the classic BTL model assumes static intrinsic scores, overlooking item-related covariates. Such oversight can lead to inefficient use of important information present in the data, especially when these covariates are available and relevant. For instance, a movie recommendation system might benefit from considering variables like genre and duration, just as a National Basketball Association (NBA) match prediction could be improved by accounting for team-specific attributes such as offensive and defensive capabilities.

Recognizing this limitation, \cite{fan2022uncertainty} introduced the covariate-assisted ranking estimation (CARE) model. This innovative approach incorporates covariate information by assuming the intrinsic score for item $i$ is the sum of two components
\begin{equation*}
\theta_i = \alpha_i^*+\boldsymbol{z}_i^{\top}\boldsymbol{\beta}^*,
\label{eqn:CARE2part}
\end{equation*}
where $\boldsymbol{z}_i\in\mathbb{R}^d$ and $\boldsymbol{\beta}^*\in\mathbb{R}^d$ are, respectively, the observed covariate vector and the coefficient vector for item $i$.  Together, the term $\boldsymbol{z}_i^{\top}\boldsymbol{\beta}^*$ captures the effects of these covariates, while $\alpha_i^*$ represents the information that cannot be explained by the covariates.  
%With this modification, model~(\ref{eqn:btl}) now becomes:
%\begin{equation*}
%    P(\text{item $i$ beats item $j$}) = \frac{e^{\alpha_i^*+\boldsymbol{z}_i^{\top}\boldsymbol{\beta}^*}}{e^{\alpha_i^*+\boldsymbol{z}_i^{\top}\boldsymbol{\beta}^*}+e^{\alpha_j^*+\boldsymbol{z}_j^{\top}\boldsymbol{\beta}^*}}.
%\end{equation*}
Below are more details about the CARE model. It has been shown that the CARE model performs better than the classic BTL model in predictions and inferences when useful covariate information is available. 

Though CARE offers a significant advancement, many real-world applications involve sequential comparisons, which means scores might evolve over time. While several methods have been developed to address this by assuming smooth temporal transitions in the BTL model \citep{cattelan2013dynamic, glickman1999parameter, glickman2001dynamic, bong2020nonparametric}, they often do not account for sudden shifts, such as an NBA team's star player's injury. A notable exception is the recent work of \cite{li2022detecting} that recognizes these abrupt changes. However, its focus was limited to the traditional BTL model, neglecting the rich covariate data.

This paper introduces a systematic approach to detect abrupt changes while simultaneously accounting for covariate information. The core concept involves partitioning the data temporally into distinct segments and fitting a separate CARE model to each. Consequently, the junctions where adjacent CARE models converge signify abrupt changes. These junctures are hereafter referred to as {\em change points}. Change point detection is a problem that attracts enormous attention today \citep[e.g.,][]{zhao2021alternating, han2024some, wang2024statistical, wang2024improvinglungcancerdiagnosis,10526478,safikhani2022fast}.  The task of estimating both the number and precise locations of these change points is non-trivial. To address this, this paper employs the minimum description length (MDL) principle \citep{rissanen1998stochastic,rissanen2007information} as an estimator. Furthermore, the PELT algorithm \citep{killick2012optimal} is harnessed for the practical estimation process. Given the past successes of both MDL and PELT in diverse change point detection challenges, it is not surprising that the proposed method exhibits both compelling theoretical and empirical strengths.

The rest of this paper is organized as follows. Section~\ref{section2} introduces the model formulation for the piecewise stationary covariate-assisted ranking estimation (PS-CARE) model. Section~\ref{section3} derives the MDL criterion for estimating the unknowns in the PS-CARE model.  It also studies the theoretical properties of the criterion. Section~\ref{section4} develops a PELT algorithm to minimize the MDL criterion.  The empirical performance of the proposed method is illustrated in Section~\ref{section5} via various numerical simulations and in Section~\ref{section6} via an application to some real NBA match data.  Lastly, concluding remarks are offered in Section~\ref{section7}, while technical details are provided in the appendix.

\section{Model Formulation}
\label{section2}
This section provides the precise definition of the PS-CARE model for change point detection.  We first describe the CARE model of \cite{fan2022uncertainty}.

\subsection{The CARE Model}
Recall that in the CARE model, the intrinsic core for item $i$ is modeled as the sum of two components~(\ref{eqn:CARE2part}): one component $\boldsymbol{z}_i^{\top}\boldsymbol{\beta}^*$ captures the contributions of the covariates and the other component $\alpha_i^*$ captures the variations that cannot be explained by the covariates. With this modification, the BTL score~(\ref{eqn:btl}) becomes:
\begin{equation*}
    P(\text{item $i$ beats item $j$}) = \frac{e^{\alpha_i^*+\boldsymbol{z}_i^{\top}\boldsymbol{\beta}^*}}{e^{\alpha_i^*+\boldsymbol{z}_i^{\top}\boldsymbol{\beta}^*}+e^{\alpha_j^*+\boldsymbol{z}_j^{\top}\boldsymbol{\beta}^*}}.
\end{equation*}
Additional constraints are required to make the CARE model identifiable. Write $\boldsymbol{\alpha}^*=(\alpha_1^*,\cdots, \alpha_n^*)$ and $\boldsymbol{Z} =[\boldsymbol{z}_1, \cdots, \boldsymbol{z}_n]^{\top}$. These additional constraints are $\sum_{i=1}^n\alpha_i^*=0$ and $\boldsymbol{Z}^{\top}\boldsymbol{\alpha}^*=0$.

Furthermore, we collect all the parameters in $\boldsymbol{\xi}=(\boldsymbol{\alpha}^*,\boldsymbol{\beta}^{*\top})^{\top}$, and denote $\boldsymbol{w_i}=[1, \boldsymbol{z}_i]^{\top}$ and $\boldsymbol{W}=[\boldsymbol{w_1},\cdots,\boldsymbol{w_n}]^{\top}$. With these, the constrained parameter set is defined as $\Theta_n=\{(\boldsymbol{\alpha},\boldsymbol{\beta}): \boldsymbol{W}^{\top}\boldsymbol{\alpha}=\boldsymbol{0}\}$.

%\tlee{use $\boldsymbol{\gamma}$ to replace $\boldsymbol{\xi}$, and something like $\boldsymbol{w_i}$ and $\boldsymbol{W_i}$ to replace $\boldsymbol{\Bar{z}_i}$ and  $\boldsymbol{\Bar{Z}_i}$}
%\yh{$\gamma$ is already used, I changed it to $\xi$}

Efficient methods for parameter estimation and uncertainty quantification for the CARE model are provided by \cite{fan2022uncertainty}.

\subsection{Piecewise Stationary CARE model}
This subsection presents the PS-CARE model for change point detection for the ranking problem, where the pairwise comparisons are done sequentially.  Briefly, a PS-CARE model is a concatenation of a sequence of different CARE models, where changes occur when one model switches to a different one.

We need some notation to proceed.  We assume $T$ pairwise comparisons are performed at $T$ distinct time points $t=1, \cdots, T$.  For any positive integer $N$, we denote $[N]$ as the set $\{1, \cdots, N\}$ containing all positive integers less than or equal to $N$.  Thus, the comparisons occurred at $t\in [T]$.  

Let $\boldsymbol{\xi}{(t)}$ denote the value of $\boldsymbol{\xi}$ at time $t$.  We assume there are $K \geq 1$ change points $\{\tau_1, \cdots, \tau_K\}$ such that the following conditions are satisfied:
\begin{itemize}
\item $1=\tau_0<\tau_1<\tau_2<\cdots<\tau_K<\tau_{K+1}=T$,
\item $\tau_k \in \{1, \cdots, T\}$ for all $k=1, \cdots, K$,
\item $\boldsymbol{\xi}{(t)}\neq \boldsymbol{\xi}(t-1)$ if $t \in \{\tau_1, \cdots, \tau_K\}$, and %$t \in \{\tau_k\}_{k=1}^{K}$, and 
\item $\boldsymbol{\xi}{(t)} = \boldsymbol{\xi}(t-1)$ if $t \notin \{\tau_1, \cdots, \tau_K\}$. %if $t\notin \{\tau_k\}_{k=1}^{K}$.
\end{itemize}

In other words, the $K$ change points $\{\tau_1, \cdots, \tau_K\}$ partition the whole time span $\{1, \cdots, T\}$ into $K+1$ segments, for which the value of $\boldsymbol{\xi}{(t)}$ remaining within each segment. The values of $K$ and $\tau_k$'s are unknown and need to be estimated.

We also denote the relative location of $\{\tau_1, \cdots, \tau_K\}$ to be $\boldsymbol{\lambda} = (\lambda_1,\cdots,\lambda_K)$, where $\lambda_k = \tau_k/T$ for $i = 1,\cdots, K$, and naturally we set $\lambda_0=0$ and $\lambda_{K+1}=1$. 
%\tlee{Should the denominator be $T$ instead of $K$? Also, it would be better to use $k$ as the subscript for the relative/change points; i.e., $\tau_k, \lambda_k$} 
Assume that at time point $t$ item $i$ and item $j$ are compared and 
\begin{equation}
\label{p_ij_CARE}
    P(y_t = 1) := P(\text{item $i$ beats item $j$ at time $t$}) =  P[y_t=1|\boldsymbol{\xi}(t)] 
 = \frac{e^{\alpha_{it}^*+\boldsymbol{z}_i^{\top}\boldsymbol{\beta}^*_t}}{e^{\alpha_{it}^*+\boldsymbol{z}_i^{\top}\boldsymbol{\beta}^*_t}+ e^{\alpha_{jt}^*+\boldsymbol{z}_j^{\top}\boldsymbol{\beta}^*_t}},
\end{equation}
where $(\alpha_{1t}^*,\alpha_{2t}^*,\cdots,\alpha_{nt}^*,\boldsymbol{\beta}_t^*)=\boldsymbol{\xi}(t)$. We further denote $\boldsymbol{\xi}_T^{(k)} := \boldsymbol{\xi}(t), \forall t\in \{\tau_{k-1}+1,\cdots,\tau_{k}\}$ (if time $t$ belongs to the $k$-th segment).  
%\tlee{same as before, use $k$ instead of $j$ as the subscript}

Define $\boldsymbol{\Tilde{z}}_i=(\boldsymbol{e}_i^{\top}, \boldsymbol{z}_i^{\top})^{\top}$, where $\{\boldsymbol{e}_i\}_{i=1}^n$ represents the canonical basis vectors in $\mathbb{R}^n$.  Let $\boldsymbol{z}_t^*=\boldsymbol{\Tilde{z}}_i - \boldsymbol{\Tilde{z}}_j$ if items $i,j\in\{1,2,\cdots,n\}$ are compared at time $t$ where $t$ belongs to the $k$-th segment or $\tau_{k-1}+1\leq t<\tau_{k}$,
then the log-likelihood function of observation $y_t$ can be written as 
\begin{equation}
\label{log_like_one_observation}
    l_k(\boldsymbol{\xi}_T^{(k)}; y_{t}, \boldsymbol{z}_t^*) := y_t\boldsymbol{z}_t^*\boldsymbol{\xi}_T^{(k)} - \log(1+\exp(\boldsymbol{z}_t^{*T}\boldsymbol{\xi}_T^{(k)})).
\end{equation}

The log-likelihood of all data is then given by 
\begin{equation}
\label{all_log_likeihood}
    L_T(\boldsymbol{y}) = \sum_{k=1}^{K+1}\sum_{t=\tau_{k-1}+1}^{\tau_k}l_k(\boldsymbol{\xi}_T^{(k)}; y_{t,k}, \boldsymbol{z}_t^*) = \sum_{k=1}^{K+1}\sum_{t=\tau_{k-1}+1}^{\tau_k}y_{t,k}\boldsymbol{z}_t^*\boldsymbol{\xi}_T^{(k)} - \log(1+\exp(\boldsymbol{z}_t^{*T}\boldsymbol{\xi}_T^{(k)})).
\end{equation}
The main goal is to estimate the number and locations of the change points as well as the parameters related to the intrinsic scores for the items. In other words, we want to estimate $K$, $\boldsymbol{\lambda}$ and $\boldsymbol{\xi}_T^{(k)}$ for $k=\{1,2,\cdots, K+1\}$. Thus, we can obtain the ranking of $n$ items in each segment and recover the PS-CARE model. In order to estimate all these parameters simultaneously, the minimum description length criterion is applied. 

%\tlee{list the parameters of a PS-CARE model}

\section{Change Point Detection using MDL}
%\section{Methodology}
\label{section3}

The MDL principle is a popular and effective method for deriving effective selection criteria for model selection problems. Instead of assuming the data is generated from a given model, it views statistical modeling as a means of generating descriptions of the observed data. Thus, it defines the best-fitting model as the one that compresses the data into the shortest possible code length for storage, where the code length can be thought of as the number of bits needed to store the observed data. The MDL principle was proposed by \cite{rissanen1998stochastic,rissanen2007information} and has been successfully applied to solve various model selection problems such as image segmentation \citep{lee2000minimum}, network constructions \citep{bouckaert1993probabilistic}, and quantile and spline regression \citep{leerobust02, aue2014segmented}.  There are various versions of MDL criteria, and this paper focuses on the so-called ``Two-Part MDL'' \citep[e.g.,][]{lee2001introduction}.  This section derives the corresponding MDL criterion for fitting a PS-CARE model.

\subsection{Derivation of the MDL Criterion}
To store the observed comparison results $\boldsymbol{y}=\{y_1,y_2\cdots,y_T\}$, the data can be divided into two parts: the first part is a fitted model and the second part is the corresponding residuals that cannot be explained by the fitted model.  If the fitted model describes the data well, it will be more economical to store the data in this way.  
%The derivation of MDL can be formulated as follows. To model the observed data, we can model it into two parts, one is the code length of the proposed models, another part is the code length of its corresponding residuals. 
Denote CL$(z)$ as the code length required to store any object $z$; thus, we want to minimize CL$(\text{``observed data''})$.  Also, denote the whole class of PS-CARE models as $\mathcal{M}$.  Lastly, denote any model in $\mathcal{M}$ as $\mathcal{F}\in\mathcal{M}$, the estimated version of $\mathcal{F}$ as $\hat{\mathcal{F}}$,  and its corresponding residuals as $\hat{\mathcal{E}}$.  Then we have
\begin{eqnarray}
\text{CL}(\text{``observed data''}) & = & 
\text{CL(``fitted model'')} + \text{CL(``residuals'')} \nonumber \\
& = & \text{CL}(\hat{\mathcal{F}}) + \text{CL}(\hat{\mathcal{E}}|\hat{\mathcal{F}}). \label{MDL_data}
\end{eqnarray}

Now, we need computable expressions for $\text{CL}(\hat{\mathcal{F}})$ and $\text{CL}(\mathcal{\hat{\mathcal{E}}}|\hat{\mathcal{F}})$ and we first calculate $\text{CL}(\hat{\mathcal{F}})$.  Notice that in order to completely determine a model $\mathcal{F}$ for a PS-CARE model, the parameters that we need to know including change points number $K$ and their locations $\boldsymbol{\tau}=\{\tau_1,\cdots,\tau_{K}\}$. In addition, for each segment $k=1, \ldots, K+1$, we need to know the intrinsic score related parameters 
$
    \boldsymbol{\xi}_T^{(k)} = (\alpha_{1,k},\alpha_{2,k},\cdots,\alpha_{n,k},\boldsymbol{\beta}_{k}), 
$
for the $k$-th segment.  Write $\boldsymbol{\hat{\xi}_T}=(\boldsymbol{\hat{\xi}}_T^{(1)},\cdots,\boldsymbol{\hat{\xi}}_T^{(K+1)})$. Then we have $\hat{\mathcal{F}} = (K,\boldsymbol{\tau},\boldsymbol{\hat{\xi}_T})$, so the first part code length for fitted model $\hat{\mathcal{F}}$ can be represented as 
\begin{equation}
\label{eqn:CLF}
    \text{CL}(\hat{\mathcal{F}}) = \text{CL}(K) + \text{CL}(\boldsymbol{\tau}) +\text{CL}(\boldsymbol{\hat{\xi}_T}).
\end{equation}

%To apply the MDL criteria to $\mathcal{F}$, we need to first encode breakpoint number $m$ and the locations $\mathcal{T}=\{\tau_1,\cdots,\tau_{m}\}$. We also need to encode the lag orders $\boldsymbol{p}_j=(p_{1,j},p_{2,j}),j=1,\cdots,m+1$ and regression parameters 
%\begin{equation}
%\label{regression_para}
%    \theta_j = (\beta_{0,j},\alpha_{1,j},\cdots,\alpha_{p_{1,j},j},\beta_{1,j},\cdots,\beta_{p_{2,j},j},\gamma_j), 
%\end{equation}
%$j=1,\cdots,m+1$ in each segment. Denote $\mathcal{P} = (\boldsymbol{p}_1,\cdots,\boldsymbol{p}_{m+1})$, $\Theta=(\theta_1,\cdots,\theta_{m+1})$ then $\mathcal{F} = (m,\mathcal{T},\mathcal{P},\Theta)$ which leads to code length decomposition:
%\begin{equation}
%    \text{CL}(\mathcal{F}) = \text{CL}(m) + \text{CL}(\mathcal{T}) + \text{CL}(\mathcal{P}) +\text{CL}(\Theta)
%\end{equation}

In general \citep{rissanen1998stochastic}, it takes about $\log(I)$ bits to encode an unknown integer $I$, and it takes $\log(I_u)$ bits to encode it if we know that it has an upper bound $I_u$. So the first two terms on the RHS of~(\ref{eqn:CLF}) are
\begin{eqnarray}
\text{CL}(K) &=& \log(K+1), 
\label{MDL_K} \\
\text{CL}(\boldsymbol{\tau}) & = & (K+1)\log(T), \label{MDL_F1}
\end{eqnarray}
where the additional $1$ in $\text{CL}(K)$ is used to make the formula meaningful when there are no change points, i.e., $K=0$. 

It remains to calculate the last term in~(\ref{eqn:CLF}). To obtain $\text{CL}(\boldsymbol{\hat{\xi}_T})$, we need to first estimate $\boldsymbol{\hat{\xi}_T}$ from model~(\ref{p_ij_CARE}) and then encode the estimated values we calculated.  For estimation, we shall use the maximum likelihood method of \cite{fan2022uncertainty}, which has been proven to possess excellent asymptotic properties.  Meanwhile, for encoding the maximum likelihood estimate we obtained, we shall use the result of \cite{rissanen1998stochastic} that any (scalar) maximum likelihood estimate calculated from $N$ observations can be effectively encoded with $\frac{1}{2}\log(N)$ bits. The maximum likelihood estimate can be obtained by running a projected gradient descent algorithm, and we shall denote the maximum likelihood estimator of $\boldsymbol{\xi}_T$ by $\boldsymbol{\hat{\xi}}_T$.
%\begin{equation*}
%\begin{aligned}
%    Z_{l(t-1),j}^{*} := \{1,w_l^{\top}\boldsymbol{X}_{t-1,j},w_l^{\top}\boldsymbol{X}_{t-2,j},\cdots,w_l^{\top}\boldsymbol{X}_{t-p_{1,j},j}, %\\ &
%    X_{l(t-1),j}, 
%    X_{l(t-2),j},\cdots,X_{l(t-p_{2,j}),j},V_l^{\top}\}^{\top} 
%    %\\     &
%    \in\mathbb{R}^{p_{1,j}+p_{2,j}+q+1},
%\end{aligned}    
%\end{equation*}

As mentioned before, to encode a scalar maximum likelihood estimate, the code length is $\frac{1}{2}\log(N)$ if $N$ observations were used for estimation.   Therefore, for the PS-CARE model, we have 
\begin{equation}
\label{MDL_MLE}
    \text{CL}(\boldsymbol{\hat{\xi}}_T) = \sum_{k=1}^{K+1}\frac{n+d-1}{2}\log(n_k),
\end{equation}
where $n_k=\tau_{k}-\tau_{k-1}$ represents the length of $k$-th segment.

The second and last part in~(\ref{MDL_data}) that we need to calculate is $\text{CL}(\hat{\mathcal{E}}|\hat{\mathcal{F}})$, which is the residuals of the fitted model $\hat{\mathcal{F}}$. It equals the negative log (base 2) of the likelihood of the fitted model $\hat{\mathcal{F}}$ \citep{rissanen1998stochastic}.  From~$(\ref{all_log_likeihood})$ we have 
\begin{equation}
\label{MDL_residual}
\begin{aligned}
    \text{CL}(\hat{\mathcal{E}}|\hat{\mathcal{F}}) &= \sum_{k=1}^{K+1}\sum_{t=\tau_{k-1}+1}^{\tau_k}\left(-y_{t,k}\boldsymbol{z}_t^{*T}\boldsymbol{\xi}_T^{(k)} + \log\left(1+\exp(\boldsymbol{z}_t^{*T}\boldsymbol{\xi}_T^{(k)})\right)\right)\log_2e.
\end{aligned}
\end{equation}

Combining~(\ref{MDL_K}), (\ref{MDL_F1}), (\ref{MDL_MLE}) and~(\ref{MDL_residual}) and using logarithm base $e$ instead of base $2$, (\ref{MDL_data}) becomes
\begin{eqnarray}
%\begin{aligned}
\text{CL}(\text{``data''}) & = & \log(K+1) + (K+1)\log(T) 
%\\ &
+ \sum_{k=1}^{K+1}\left(\frac{n+d-1}{2}\log(n_k) \right) \nonumber \\ 
& & + \sum_{k=1}^{K+1}\sum_{t=\tau_{k-1}+1}^{\tau_k}\left(-y_{t,k}\boldsymbol{z}_t^{*T}\boldsymbol{\xi}_T^{(k)} + \log\left(1+\exp(\boldsymbol{z}_t^{*T}\boldsymbol{\xi}_T^{(k)})\right)\right)\log_2e \nonumber \\
&:= & \text{MDL}(K,\boldsymbol{\tau},\boldsymbol{\xi}_T). \label{MDL_all}
%\end{aligned}
\end{eqnarray}
Thus, the MDL principle suggests that the best-fitting PS-CARE model for the observed data $\boldsymbol{y}=\{y_1,y_2\cdots,y_T\}$, is the one $\hat{\mathcal{F}}\in\mathcal{M}$ that minimizes $(\ref{MDL_all})$.

\subsection{Theoretical Properties}
Denote the true number of change points as $K_0$ and the true locations of the change points as $\boldsymbol{\tau_0} = \{\tau_1^0,\cdots,\tau_{K_0}^0\}$.  Define the true relative change points locations as $\boldsymbol{\lambda}_0 = \{\lambda_1^0,\cdots,\lambda_{K_0}^0\}$ with $\tau_k^0=\lfloor \lambda_k^0T \rfloor$ for $k=1, \ldots, K_0$, where $\lfloor x \rfloor$ represents the greatest integer that is less than or equal to $x$. Note that the theoretical results in this subsection will be presented in terms of $\boldsymbol{\lambda}$ instead of $\boldsymbol{\tau}$ since they are equivalent.  

As suggested by \cite{davis2006structural}, for each segment, a sufficient number of comparisons are required to estimate the corresponding CARE model parameters adequately.  For this reason, we impose the following constraint on the estimate of $\boldsymbol{\lambda}$.  First, choose $\epsilon_\lambda>0$ sufficiently small enough that $\epsilon_\lambda\ll\min_{k=1,\cdots,K_0+1}(\lambda_k^0-\lambda_{k-1}^0)$. Then define 
%\begin{equation}
%\label{lambda_assump}
%A^{K}_{\epsilon_\lambda}=\left\{\left(\lambda_{1}, \ldots, \lambda_{K}\right), 0=\lambda_0<\lambda_{1}<\cdots<\lambda_{K}<\lambda_{K+1}=1, \lambda_{k}-\lambda_{k-1} \geq \epsilon_\lambda, k=1,2, \ldots, K+1\right\},
%\end{equation}
\begin{multline}
\label{lambda_assump}
A^{K}_{\epsilon_\lambda}=\left\{\left(\lambda_{1}, \ldots, \lambda_{K}\right), 0=\lambda_0<\lambda_{1}<\cdots<\lambda_{K}<\lambda_{K+1}=1, \right. \\
\left. \lambda_{k}-\lambda_{k-1} \geq \epsilon_\lambda, k=1,2, \ldots, K+1\right\},
\end{multline}
so we require the estimate of $\boldsymbol{\lambda}$ to be an element of $A^{K}_{\epsilon_\lambda}$. Under this constraint, the number of change points is also bounded by $M=\ceil*{1/\epsilon_\lambda}+1$. As for coefficient $\boldsymbol{\xi}_T$, its constraint set is $\Theta_n = \{(\boldsymbol{\alpha},\boldsymbol{\beta}): \boldsymbol{W}\boldsymbol{\alpha}=\boldsymbol{0}\}$, which guarantees the model identifiability.

To obtain the maximum likelihood estimator with desirable properties for the CARE model, several other assumptions are required \citep{fan2022uncertainty}.

\begin{assumption}
\label{assumption_1}
Denote the projected matrix as $\mathcal{P}_{\boldsymbol{W}}:=\boldsymbol{W}(\boldsymbol{W}^{\top}\boldsymbol{W})^{-1}\boldsymbol{W}^{\top}$. Assume that there exists a positive constant $c_0$ such that 
\begin{equation}
    \left\|\mathcal{P}_{\boldsymbol{W}}\right\|_{2, \infty}=\left\|\boldsymbol{W}\left(\boldsymbol{W}^{\top} \boldsymbol{W}\right)^{-1} \boldsymbol{W}^{\top}\right\|_{2, \infty} \leq c_0 \sqrt{(d+1) / n}.
\end{equation}
\end{assumption}

Assumption \ref{assumption_1} is called the incoherence condition that guarantees the rows of $\mathcal{P}_{\mathbf{W}}$ to be nearly balanced and the row sum of squares all of order $(d+1)/n$ or smaller.

\begin{assumption}
    
\label{assumption_2}
Denote $\boldsymbol{\Tilde{z}}_i=(\boldsymbol{e}_i^{\top}, \boldsymbol{z}_i^{\top})^{\top}$, where $\{\boldsymbol{e}_i\}_{i=1}^n$ represents the canonical basis vectors in $\mathbb{R}^n$ and  $\boldsymbol{\Sigma}=\sum_{i>j}\left(\boldsymbol{\Tilde{z}}_i-\boldsymbol{\Tilde{z}}_j\right)\left(\boldsymbol{\Tilde{z}}_i-\boldsymbol{\Tilde{z}}_j\right)^{\top}$. Assume there exists positive constants $c_1$ and $c_2$ such that 
\begin{equation}
    c_2 n \leq \lambda_{\min , \perp}(\boldsymbol{\Sigma}) \leq\|\Sigma\|_{op} \leq c_1 n,
\end{equation}
where $\|\Sigma\|_{op}$ is the operator norm of $\boldsymbol{\Sigma}$ and 
\begin{equation}
    \lambda_{\min , \perp}(\boldsymbol{\Sigma}):=\min \left\{\boldsymbol{u} \mid \boldsymbol{u}^{\top} \boldsymbol{\Sigma} \boldsymbol{u} \geq \mu\|\boldsymbol{u}\|_2^2 \text { for all } \boldsymbol{u} \in \Theta_n\right\}.
\end{equation}
\end{assumption}

Assumption \ref{assumption_2} constraints the covariance matrix $\boldsymbol{\Sigma}$ to be well-behaved in all directions inside the parameter space $\Theta_n$ by restricting its largest and smallest eigenvalues are both of order $n$.

Now, we introduce a connected graph notion that describes the sampling scheme for collecting comparison data over time. Following \cite{li2022detecting}, we consider a connected comparison graph $\mathcal{G}=\mathcal{G}([n], E)$ with edge set $E \subseteq E_{\text {full }}:=\{(i, j):1 \leq i<j \leq n\}$.  At each time point $t \in [T]$, an element $(i_t, j_t) \in [n]\times[n]$ is randomly selected from the edge set $E$ to determine which two items are to be compared. This selection process is independent over time.  In our work, we do not require the graph to be fully connected. That is, we do not require every item to be compared with every other item.  For a specific time interval $\mathcal{I}$, we define similarly a random comparison graph $\mathcal{G_I}(V_\mathbf{I},E_\mathbf{I})$ with vertex set $V:=[n]$ and edge set $E_{\mathcal{I}}:=\{(i, j): i \text { and } j \text { are compared in } \mathcal{I}\} \subset E$.  This graph notation will be useful in studying the theoretical properties of the proposed method and will be first used by Assumption~\ref{assumption_3} below.

%Define the notation $a_n\lesssim b_n$ represents $a_n=O(b_n)$ which means for non-negative sequences $\{a_n\}$ and $\{b_n\}$, there exists a constant $v_1$ such that $a_n\leq v_1b_n$.

%Below we introduce the sampling scheme for collecting pairwise comparison data over time. For a specific time interval $\mathcal{I}$, we define a random comparison graph $\mathcal{G_I}(V_\mathbf{I},E_\mathbf{I})$ with vertex set $V:=[n]$ and edge set $E_{\mathcal{I}}:=\{(i, j): i \text { and } j \text { are compared in } \mathcal{I}\} \subset E$. 

\begin{assumption}
\label{assumption_3}
In each segment $I_k$, $k = 1,\cdots,K+1$, suppose $L_{g,h,k}$ represents the number of times items $g$ and $h$ compared in segment $I_k$, $g,h\in\{1,2,\cdots,n\}$, thus $\sum_{g,h\in\mathcal{G}_k}L_{g,h,k}=t_k$, where $t_k$ represents the time points of segment $I_k$ and $\mathcal{G}_k$ represents the graph of segment $I_k$. Let $L_{min,k}=\min_{g,h\in[n]^2}(L_{g,h,k})$. It is required $L_{min,k}\leq c_1\cdot n^{c_2}$ for any absolute constants $c_1, c_2>0$.  
Also, it is required that $n\cdot \frac{L_{min,k}}{t_k}>c_p\log(n)$ for some $c_p>0$ and $d+1<n$, $(d+1)\log(n) \lesssim n\cdot \frac{L_{min,k}}{t_k}$, where the notation $a_n\lesssim b_n$ denotes $a_n=O(b_n)$.

%Also, it is required that $n\cdot \frac{L_{min,k}}{t_k}>c_p\log(n)$ for some $c_p>0$ and $d+1<n$, $(d+1)\log(n) \lesssim n\cdot \frac{L_{min,k}}{t_k}$. 

%\tlee{1: I don't quite understand the last sentence.  2: since the notation "$\lesssim$" is only used once, we can change the last sentence to one of the following three possibilities.  Please let me know what you think.} {\color{blue}{This sentence is required to make sure that each segment has enough comparisons to form a consistent MLE estimation of the intrinsic score. Some of the requirements are written here to match the requirement of a consistent MLE estimator's existence in Jianqing Fan's paper. I think the second last one looks good.}}

%Also, it is required that $n\cdot \frac{L_{min,k}}{t_k}>c_p\log(n)$ for some $c_p>0$ and $d+1<n$, $(d+1)\log(n) = O(n\cdot \frac{L_{min,k}}{t_k})$. 

%Also, it is required that $n\cdot \frac{L_{min,k}}{t_k}>c_p\log(n)$ for some $c_p>0$ and $d+1<n$, $(d+1)\log(n) \lesssim n\cdot \frac{L_{min,k}}{t_k}$, where the notation $a_n\lesssim b_n$ denotes $a_n=O(b_n)$.

%Also, it is required that $n\cdot \frac{L_{min,k}}{t_k}>c_p\log(n)$ for some $c_p>0$ and $d+1<n$, $(d+1)\log(n) \lesssim n\cdot \frac{L_{min,k}}{t_k}$, where the notation $a_n\lesssim b_n$ denotes $a_n=O(b_n)$, which means for non-negative sequences $\{a_n\}$ and $\{b_n\}$, there exists a constant $v_1$ such that $a_n\leq v_1b_n$.

\end{assumption}

Assumption \ref{assumption_3} guarantees in each segment, the number of comparisons for every two items should meet some lower bound related to the covariate dimensions and item numbers. This assumption also guarantees the connectivity of graph $\mathcal{G}_k$ as well as the consistency of the maximum likelihood estimator in each segment \citep{fan2022uncertainty}.

% \begin{equation*}
%     \begin{aligned}
% A_{m}=\left\{(\lambda_1,\cdots,\lambda_m), 0=\lambda_{0}<\lambda_{1}<\lambda_{2}<\cdots<\lambda_{m}<\lambda_{m+1}=1\right.\\
% \left.\lambda_{i}-\lambda_{i-1} \geq \xi, i=1,2, \ldots, m+1\right\}
% \end{aligned}
% \end{equation*}
%Denote parameters $\boldsymbol{p}:=(p_{1,1},p_{2,1},\cdots,p_{1,m+1},p_{2,m+1})$, $m$ and  $\boldsymbol{\lambda}=(\lambda_1,\cdots,\lambda_m)$ are estimated by minimizing MDL $(\ref{MDL_all})$ over $m\leq M_0$, $0 \leq \boldsymbol{p}\leq P_0$, and $\boldsymbol{\lambda}\in A_{m}$ where $M_0$ and $P_0$ are prespecified upper bounds,

Using this assumption and ($\ref{MDL_all}$), the unknown parameters are given by
\begin{equation}
\label{MDL_estimate}
    \{\hat{K},\boldsymbol{\hat{\lambda}_T},\boldsymbol{\hat{\xi}}_T\} = \argmin_{K\leq M, \boldsymbol{\xi}_T\in\Theta_n, \boldsymbol{\lambda}\in A^K_{\epsilon_\lambda}}\frac{2}{T}\text{MDL}(K,\boldsymbol{\lambda},\boldsymbol{\xi}_T).
\end{equation}

\begin{thm}
\label{Theorem1}
For the PS-CARE model given by~$(\ref{p_ij_CARE})$, denote the true number of change points as $K_0$ and the true relative locations of change points as $\boldsymbol{\lambda}^0$. The estimate $\hat{\boldsymbol{\lambda}}$ defined by~(\ref{MDL_estimate}) satisfies
\begin{equation}
\label{Theorem 1}
    \hat{K} \stackrel{P}\longrightarrow K_0, \quad \boldsymbol{\hat{\lambda}_T} \stackrel{P} \longrightarrow \boldsymbol{\lambda}^0.
\end{equation}
\end{thm}

Theorem \ref{Theorem1} shows that the MDL criterion enjoys some desirable consistency properties. The detailed proof of Theorem~\ref{Theorem1} can be found in the appendix.

\section{Practical Optimization of MDL}
\label{section4}
Due to the huge search space, it is very challenging to locate the global minimum of (\ref{MDL_estimate}). Here, we propose solving this optimization problem using the PELT algorithm of \cite{killick2012optimal}. It has been shown that, for a class of change point detection problems, the PELT algorithm is an exact search algorithm with linear computational cost, which makes it extremely appealing in practice.

The objective MDL criterion (\ref{MDL_all}) can be rewritten as
\begin{equation}
\label{pelm_obj}
    \sum_{k=1}^{K+1}\left[\mathcal{C}(y_{(\tau_{k-1}+1):\tau_{k}}) \right] + \gamma f(K),
\end{equation}
where 
%\begin{equation*}
%\begin{aligned}
%    \mathcal{C}(y_{(\tau_{k-1}+1):\tau_{k}}) = & \left(\frac{n+d-1}{2}\log(n_k) \right) + \\
%    & \sum_{t=\tau_{k-1}+1}^{\tau_k}(-y_{t,k}\boldsymbol{z}_t^{*\top}\boldsymbol{\xi}_T^{(k)} + %\log(1+\exp(\boldsymbol{z}_t^{*\top}\boldsymbol{\xi}_T^{(k)})))\log_2e \nonumber 
%\end{aligned}
%\end{equation*}
\begin{multline*}
        \mathcal{C}(y_{(\tau_{k-1}+1):\tau_{k}}) = \left(\frac{n+d-1}{2}\log(n_k) \right) + 
     \sum_{t=\tau_{k-1}+1}^{\tau_k}(-y_{t,k}\boldsymbol{z}_t^{*\top}\boldsymbol{\xi}_T^{(k)} +  \log(1+\exp(\boldsymbol{z}_t^{*\top}\boldsymbol{\xi}_T^{(k)})))\log_2e \nonumber 
\end{multline*}
represents the cost function for a segment and $\gamma f(K) = \log(T)(K+1)$ is the remaining penalty part. 

In order to apply the PELT algorithm, one assumption needs to be satisfied: there is a constant $R$ such that, for all $t<s<T$, 
\begin{equation}
\label{pelt_ass1}
    \mathcal{C}(y_{(t+1):s}) + \mathcal{C}(y_{(s+1):T}) +R \leq \mathcal{C}(y_{(t+1):T}).
\end{equation} 
In our case, we can choose $R=\frac{n+d-1}{2}\log(\frac{8*\pi}{T})$.  We refer the readers to \cite{killick2012optimal} for the full description of the general version of the PELT algorithm, while the version that is tailored for the current problem can be found in Algorithm~\ref{alg1}.  Following the notations of \cite{killick2012optimal}, in Algorithm~\ref{alg1} we use $F(s)$ to denote the minimization from $(\ref{pelm_obj})$ for data $y_{1:s}$, $cp(s)$ to denote the estimated change point set for time point $\{1:s\}$, and $R_{s}$ to denote the time points that are possible to be the last change points prior to $s$. 

%Let $F(s)$ denote the minimization from $(\ref{pelm_obj})$ for data $y_{1:s}$ and $cp(s)$ represent the estimated change point set for time point $\{1:s\}$; $R_{s}$ represent the time points that are possible to be the last breakpoints prior to $s$. 
%\tlee{are these notations from the PELT algorithm?} \yh{Yes, from the original PELT paper}

%So, we can now use the PELT algorithm to estimate the locations and numbers of breakpoints simultaneously; see Algorithm~\ref{alg1}

\begin{algorithm}
%\textsl{}\setstretch{1.8}
        \setstretch{1.35}
	\renewcommand{\algorithmicrequire}{\textbf{Input:}}
	\renewcommand{\algorithmicensure}{\textbf{Output:}}
	\caption{Optimize MDL criterion based on the PELT Algorithm}
	\label{alg1}
	\begin{algorithmic}[1]
 
		  \REQUIRE { A set of observed pairwise comparison data $(y_1, y_2, \cdots, y_T)$, $(\boldsymbol{z}_1^*, \boldsymbol{z}_2^*, \cdots, \boldsymbol{z}_T^*)$. \\
                \quad \quad A prespecified constant $R$ satisfied (\ref{pelt_ass1}) \\
                \quad \quad A minimum length $L$ for each estimated segments.}

        \STATE Initialization: Set $F(0)=\cdots=F(L-1)=-\gamma$; $F(i)=\mathcal{C}(y_{1:i}), \forall i=L,L+1,\cdots,2L-1$; $R_1=R_2=\cdots=R_{2L-1}=\{0\}, R_{2L}=\{0,L\}$.
        \FOR{$\tau^*=2L,\cdots,T$}
        \STATE \quad Calculate $F(\tau^*)=\min_{\tau\in R_{\tau^*}}\left[F(\tau)+\mathcal{C}(y_{(\tau+1):\tau^*})+\gamma\right]$.
        \STATE \quad Let $\tau^1 =\arg\{\min_{\tau\in R_{\tau^*}}\left[F(\tau)+\mathcal{C}(y_{(\tau+1):\tau^*})+\gamma\right]\}$.
        \STATE \quad Set $cp(\tau^*)=[cp(\tau^1),\tau^1]$.
        \STATE \quad Set $R_{\tau^*+1}=\{\tau^* \cap \{\tau\in R_{\tau^*}: F(\tau)+\mathcal{C}(y_{(\tau+1):\tau^*})+\gamma < F(\tau^*) \} \}$.
        \ENDFOR
        \ENSURE {The change points recorded in $cp(T)$}.
	\end{algorithmic}  
\end{algorithm}

\section{Simulation Results}
\label{section5}
In this section, we report the numerical results of our proposed method in different simulation settings. We follow a similar convention as in \cite{li2022detecting}. Suppose we have $K$ change points $\{\tau_k\}_{k\in[K]}$ in the sequential comparison data, $\tau_0=1$. Suppose the number of objects is $n$ and the dimension of covariates is $d$. For each setting, we set the comparison graph $\mathcal{G_I}(V_\mathbf{I},E_\mathbf{I})$ to be the complete graph and $T = (K+1)\Delta$ with the true change point located at $\tau_i = i\Delta$ for $i\in[K]$. To generate the covariates $\boldsymbol{Z}$ and the coefficients $\{\alpha^*\}$ and $\boldsymbol{\beta}^*$, we follow the same idea as in \cite{fan2022uncertainty}. The covariates are generated independently with $(\boldsymbol{z}_i)_j\sim $ Uniform$[-0.5, 0.5]$ for all $i\in[n], j\in[d]$. For matrix $\boldsymbol{Z}=[\boldsymbol{z}_1,\boldsymbol{z}_2\cdots,\boldsymbol{z}_n]^{\top}\in\mathbb{R}^{n\times d}$, column-wise normalization is applied and then scale $\boldsymbol{x}_i$ by $\boldsymbol{x}_i/h$ so that $\max_{i\in[n]}\|\boldsymbol{x}_i\|_2/h=\sqrt{(d+1)/n}$. Such initialization guarantees the Assumptions~\ref{assumption_1}  and~\ref{assumption_2} are satisfied. For $\boldsymbol{\alpha}^*$, its elements are generated uniformly from Uniform$[-0.3, 0.3]$. For $\boldsymbol{\beta}^*$, it is generated uniformly from a hypersphere $\{\beta, \|\beta\|_2 = 0.5\sqrt{n/(d+1)}\}$. Then we projected $\boldsymbol{\xi}$ onto the linear space $\Theta_n$.

\subsection{Setting 1: Without Covariates}
In the first simulation setting, we compare our methods to the DPLR method \citep{li2022detecting} when $d=0$, which means no covariates exist. We set $K=3$, $n=\{10,15,20,30,40,50\}$ and $\Delta=\{400,500,650\}$. Using the procedure described above, we simulate the pairwise comparison data 100 times and apply both our proposed MDL method and DPLR method to detect the number and locations of change points.

% \usepackage{booktabs}
% \usepackage{multirow}
\begin{table}[!htbp]
% \caption{Comparison of MDL and DPLR under different simulation settings without covariates. }
% \label{sim_setting_1}
\begin{center}
\begin{tabular}{@{}cccccc@{}}
\toprule
Methods & $n$                   & $\Delta$                &  mean of $\hat{\boldsymbol{\tau}}$                     & s.e. of $\hat{\boldsymbol{\tau}}$                   &  \% of  $\hat{K}=K$\\ \midrule
MDL     & \multirow{2}{*}{10} & \multirow{2}{*}{400} & (400.0, 798.2, 1200.0)  & (1.2, 1.0, 0.9)    & 98\%       \\
DPLR    &                     &                      & (398.9, 799.2, 1199.8)  & (2.0, 3.2, 1.9)  & 85\%       \\ \midrule
MDL     & \multirow{2}{*}{15} & \multirow{2}{*}{400} & (400.0, 798.0, 1200.0)  & (1.0, 1.0, 0.9)      & 99\%       \\
DPLR    &                     &                      & (402.3, 800.4, 1196.2)  & (2.3, 2.4, 2.4)  & 87\%       \\ \midrule
MDL     & \multirow{2}{*}{20} & \multirow{2}{*}{400} & (401.5, 798.2, 1200.3)  & (0.9, 1.3, 1.7)    & 100\%      \\
DPLR    &                     &                      & (399.2, 797.5, 1199.3)  & (2.6, 3.3, 2.7)  & 80\%       \\ \midrule
MDL     & \multirow{2}{*}{30} & \multirow{2}{*}{500} & (398.5, 804.1, 1199.6)  & (0.8, 1.6, 0.9)    & 96\%       \\
DPLR    &                     &                      & (392.7, 801.9, 1201.8)  & (2.6, 2.4, 3.1) & 69\%       \\ \midrule
MDL     & \multirow{2}{*}{40} & \multirow{2}{*}{500} & (499.1, 1000.1, 1500.1) & (0.9, 1.0, 1.4)    & 99\%       \\
DPLR    &                     &                      & (501.1, 999.8, 1501.1)  & (2.5, 2.6, 2.6)  & 90\%       \\ \midrule
MDL     & \multirow{2}{*}{50} & \multirow{2}{*}{650} & (649.0, 1298.6, 1950.5) & (1.1, 1.5, 1.4)  & 99\%       \\
DPLR    &                     &                      & (646.1, 1294.3, 1951.1) & (2.5, 2.6, 1.7)  & 93\%       \\ \bottomrule
\end{tabular}
\caption{Comparison of MDL and DPLR under different simulation settings without covariates.}
\label{sim_setting_1}
\end{center}
\end{table}

The results are summarized in Table~\ref{sim_setting_1}.  When calculating the means and standard errors of the estimated change points, we only consider the cases where the true number of change points was estimated.  From Table~\ref{sim_setting_1}, we can observe that for the non-covariate settings, our proposed MDL method outperforms the DPLR method. The reason might be that the DPLR method requires precise choices of some tuning parameters that might be hard to obtain, especially when the sample size ($\Delta$) is not large enough.

\subsection{Setting 2: With Covariates}
In this simulation setting, the parameters are $n=\{10, 15, 20, 30, 40, 50\}$, $d=5$, and $\Delta=\{700,800,1100,1300,1300,2000\}$. As before, we simulate the pairwise comparison data 100 times and apply both the MDL and DPLR methods to detect the number and locations of change points. Notice that the DPLR method was not designed to incorporate covariate information, and hence no such information was fed to it.

The results are summarized in Table~\ref{sim_setting_2}.  As expected, given that the proposed MDL method takes the covariate information into account, it produces superior performance when compared to the DPLR method.

% \usepackage{booktabs}
% \usepackage{multirow}
\begin{table}[!htbp]
\begin{center}
\begin{tabular}{@{}ccccccc@{}}
\toprule
Methods & $n$ & $d$                & $\Delta$                &  mean of $\hat{\boldsymbol{\tau}}$                     & s.e. of $\hat{\boldsymbol{\tau}}$                   &  \% of  $\hat{K}=K$\\ \midrule
%Methods & n                   & d                  & $\Delta$                 & $E(\boldsymbol{\tau})$                    & $SD(\boldsymbol{\tau})$                  & Percentage of  $\hat{K}=K$ \\ \midrule
MDL     & \multirow{2}{*}{10} & \multirow{2}{*}{5} & \multirow{2}{*}{700}  & (699.2, 1400.6, 2099.2)  & (0.6, 0.7, 0.4)     & 96\%       \\
DPLR    &                     &                    &                       & (697.9, 1398.6, 2096.2)  & (1.3, 1.5, 1.6)   & 60\%       \\ \midrule
MDL     & \multirow{2}{*}{15} & \multirow{2}{*}{5} & \multirow{2}{*}{800}  & (799.0, 1600.4, 2399.4)  & (0.9, 0.5, 0.5)     & 97\%       \\
DPLR    &                     &                    &                       & (798.3, 1602.4, 2402.2)  & (1.9, 1.3, 1.4)  & 66\%       \\ \midrule
MDL     & \multirow{2}{*}{20} & \multirow{2}{*}{5} & \multirow{2}{*}{1100} & (1100.0, 2200.2, 3300.3) & (0.8, 0.7, 0.6)     & 97\%       \\
DPLR    &                     &                    &                       & (1099.2, 2198.5, 3295.3) & (1.1, 1.2, 2.2)  & 84\%       \\ \midrule
MDL     & \multirow{2}{*}{30} & \multirow{2}{*}{5} & \multirow{2}{*}{1100} & (1299.8, 2599.7, 3901.0) & (0.5, 0.5, 0.6)     & 96\%       \\
DPLR    &                     &                    &                       & (1299.0, 2599.0, 3898.3) & (1.1, 1.2, 1.4) & 89\%       \\ \midrule
MDL     & \multirow{2}{*}{40} & \multirow{2}{*}{5} & \multirow{2}{*}{1300} & (1299.5, 2599.8, 3899.8) & (0.6, 0.5, 0.6)     & 98\%       \\
DPLR    &                     &                    &                       & (1301.2, 2601.6, 3903.9) & (1.5, 1.4, 1.4)    & 94\%       \\ \midrule
MDL     & \multirow{2}{*}{50} & \multirow{2}{*}{5} & \multirow{2}{*}{2000} & (2000.8, 3999.5, 5998.1) & (0.8, 0.8, 0.7)     & 97\%       \\
DPLR    &                     &                    &                       & (2001.1, 4000.3, 6000.1) & (1.7, 1.8, 1.8)  & 95\%       \\ \bottomrule
\end{tabular}
\caption{Comparison of MDL and DPLR under different simulation settings with covariates. }
\label{sim_setting_2}
\end{center}
\end{table}

% apply the DPLR method, we don't use the covariate matrix since it is designed for the vanilla BTL model instead of CARE model. The results shows the superior performance of our proposed MDL method compared with the DPLR method when there are available covariate information.

\section{Real Data Analysis}
\label{section6}
In this section, we study the game records of the NBA data as in \cite{li2022detecting}. The NBA season starts in October and ends in April the next year. One season is usually referred to as a two-year span. We focus on the same time range as in \cite{li2022detecting}, which contains records from season 1980-1981 to season 2015-2016 with 24 teams that were founded before 1990.  It has been shown in \cite{li2022detecting} that the data is non-stationary and contains multiple change points. We utilize the overall mean salary, the mean 3-point shoot percentage, and the mean rebound number of each team over the selected period as exogenous covariates.  These 3 covariates respectively represent, to a certain extent, the investment, attack ability, and defense ability of each team.

We then apply the MDL method to detect the number and locations of change points. We utilize the even-time point matches as the training dataset and the odd-time point matches as the test dataset.  
%The loss function is defined as (\ref{all_log_likeihood}). \tlee{why mention the loss functions?}  
The results are summarized in Table~\ref{NBA_tableall}.

% \usepackage{booktabs}
\begin{table}[!htbp]
\centering
\scalebox{0.7}{
\begin{tabular}{@{}llllllllll@{}}
\toprule
\multicolumn{2}{c}{\textbf{S1980-S1985}} &
  \multicolumn{2}{c}{\textbf{S1986-S1989}} &
  \multicolumn{2}{c}{\textbf{S1990-S1994}} &
  \multicolumn{2}{c}{\textbf{S1995-S1997}} &
  \multicolumn{2}{c}{\textbf{S1998-S2003}} \\ \midrule
\textbf{Celtics} &
  1.0921 &
  \textbf{Lakers} &
  1.2076 &
  \textbf{Bulls} &
  0.8809 &
  \textbf{Jazz} &
  1.1172 &
  \textbf{Spurs} &
  0.8749 \\
\textbf{76ers} &
  0.9512 &
  \textbf{Pistons} &
  0.9508 &
  \textbf{Spurs} &
  0.6787 &
  \textbf{Bulls} &
  0.9546 &
  \textbf{Lakers} &
  0.8284 \\
\textbf{Lakers} &
  0.7574 &
  \textbf{Celtics} &
  0.8223 &
  \textbf{Suns} &
  0.6489 &
  \textbf{Heat} &
  0.8756 &
  \textbf{Kings} &
  0.6981 \\
\textbf{Bucks} &
  0.7534 &
  \textbf{Trail Blazers} &
  0.6566 &
  \textbf{Jazz} &
  0.6255 &
  \textbf{Lakers} &
  0.8113 &
  \textbf{Mavericks} &
  0.5416 \\
\textbf{Nuggets} &
  0.0798 &
  \textbf{Bulls} &
  0.5303 &
  \textbf{Knicks} &
  0.5510 &
  \textbf{Trail Blazers} &
  0.5792 &
  \textbf{Timberwolves} &
  0.4022 \\
\textbf{Trail Blazers} &
  0.0626 &
  \textbf{Jazz} &
  0.4991 &
  \textbf{Rockets} &
  0.5184 &
  \textbf{Hornets} &
  0.4965 &
  \textbf{Trail Blazers} &
  0.3938 \\
\textbf{Suns} &
  0.0551 &
  \textbf{Bucks} &
  0.3841 &
  \textbf{Trail Blazers} &
  0.4513 &
  \textbf{Pacers} &
  0.4329 &
  \textbf{Jazz} &
  0.3569 \\
\textbf{Spurs} &
  0.0536 &
  \textbf{Mavericks} &
  0.3534 &
  \textbf{Cavaliers} &
  0.3044 &
  \textbf{Knicks} &
  0.4063 &
  \textbf{Pacers} &
  0.2981 \\
\textbf{Nets} &
  0.0311 &
  \textbf{76ers} &
  0.3311 &
  \textbf{Pacers} &
  0.1609 &
  \textbf{Rockets} &
  0.3865 &
  \textbf{Suns} &
  0.1213 \\
\textbf{Pistons} &
  -0.033 &
  \textbf{Suns} &
  0.2723 &
  \textbf{Lakers} &
  0.1603 &
  \textbf{Cavaliers} &
  0.3706 &
  \textbf{76ers} &
  0.0730 \\
\textbf{Knicks} &
  -0.1204 &
  \textbf{Rockets} &
  0.2544 &
  \textbf{Magic} &
  0.1531 &
  \textbf{Magic} &
  0.2291 &
  \textbf{Hornets} &
  0.0625 \\
\textbf{Rockets} &
  -0.1816 &
  \textbf{Cavaliers} &
  0.1816 &
  \textbf{Warriors} &
  0.0203 &
  \textbf{Pistons} &
  0.2100 &
  \textbf{Pistons} &
  -0.0058 \\
\textbf{Jazz} &
  -0.2818 &
  \textbf{Nuggets} &
  0.1209 &
  \textbf{Celtics} &
  -0.0237 &
  \textbf{Suns} &
  0.2063 &
  \textbf{Bucks} &
  -0.0401 \\
\textbf{Bulls} &
  -0.2942 &
  \textbf{Knicks} &
  0.0379 &
  \textbf{Hornets} &
  -0.1150 &
  \textbf{Timberwolves} &
  0.0199 &
  \textbf{Rockets} &
  -0.0841 \\
\textbf{Mavericks} &
  -0.2974 &
  \textbf{Pacers} &
  -0.0352 &
  \textbf{Nets} &
  -0.2435 &
  \textbf{Spurs} &
  -0.0685 &
  \textbf{Heat} &
  -0.1403 \\
\textbf{Kings} &
  -0.2989 &
  \textbf{Spurs} &
  -0.0702 &
  \textbf{Heat} &
  -0.25397 &
  \textbf{Bucks} &
  -0.1756 &
  \textbf{Knicks} &
  -0.1477 \\
\textbf{Warriors} &
  -0.4193 &
  \textbf{Warriors} &
  -0.0789 &
  \textbf{Pistons} &
  -0.28443 &
  \textbf{Nets} &
  -0.4655 &
  \textbf{Nets} &
  -0.1539 \\
\textbf{Pacers} &
  -0.5259 &
  \textbf{Kings} &
  -0.7207 &
  \textbf{Nuggets} &
  -0.32888 &
  \textbf{76ers} &
  -0.6068 &
  \textbf{Magic} &
  -0.2635 \\
\textbf{Clippers} &
  -0.6087 &
  \textbf{Nets} &
  -0.8198 &
  \textbf{Clippers} &
  -0.39468 &
  \textbf{Kings} &
  -0.6594 &
  \textbf{Celtics} &
  -0.2675 \\
\textbf{Cavaliers} &
  -0.7531 &
  \textbf{Timberwolves} &
  -0.9295 &
  \textbf{Kings} &
  -0.41518 &
  \textbf{Warriors} &
  -0.7487 &
  \textbf{Nuggets} &
  -0.4180 \\
\textbf{Heat} &
  NA &
  \textbf{Clippers} &
  -0.9662 &
  \textbf{Bucks} &
  -0.58409 &
  \textbf{Celtics} &
  -0.8135 &
  \textbf{Clippers} &
  -0.5736 \\
\textbf{Hornets} &
  NA &
  \textbf{Magic} &
  -0.9875 &
  \textbf{76ers} &
  -0.80124 &
  \textbf{Clippers} &
  -0.8698 &
  \textbf{Cavaliers} &
  -0.6686 \\
\textbf{Magic} &
  NA &
  \textbf{Hornets} &
  -1.0816 &
  \textbf{Mavericks} &
  -0.96357 &
  \textbf{Mavericks} &
  -1.1372 &
  \textbf{Warriors} &
  -0.7081 \\
\textbf{Timberwolves} &
  NA &
  \textbf{Heat} &
  -1.1920 &
  \textbf{Timberwolves} &
  -0.98895 &
  \textbf{Nuggets} &
  -1.5718 &
  \textbf{Bulls} &
  -1.1913 \\ \bottomrule

%\end{tabular}}
%%\caption{}
%\label{NBA_table1}
%\end{table}
&&&&&&&&&\\
%\vspace{1cm} % <--- this is new
%\begin{table}[!h]
%\scalebox{0.7}{
%\begin{tabular}{@{}llllllllll@{}}

\toprule
\multicolumn{2}{c}{\textbf{S2004-S2006}} &
  \multicolumn{2}{c}{\textbf{S2007-S2009}} &
  \multicolumn{2}{c}{\textbf{S2010-S2011m}} &
  \multicolumn{2}{c}{\textbf{S2011m-S2013}} &
  \multicolumn{2}{c}{\textbf{S2014-S2015}} \\ \midrule
\textbf{Spurs} &
  0.9037 &
  \textbf{Lakers} &
  0.9619 &
  \textbf{Bulls} &
  1.0465 &
  \textbf{Spurs} &
  0.9457 &
  \textbf{Warriors} &
  1.6399 \\
\textbf{Mavericks} &
  0.8443 &
  \textbf{Celtics} &
  0.8124 &
  \textbf{Spurs} &
  0.8086 &
  \textbf{Clippers} &
  0.8324 &
  \textbf{Spurs} &
  1.2530 \\
\textbf{Suns} &
  0.8190 &
  \textbf{Cavaliers} &
  0.7386 &
  \textbf{Heat} &
  0.6899 &
  \textbf{Heat} &
  0.7754 &
  \textbf{Clippers} &
  0.8556 \\
\textbf{Pistons} &
  0.7526 &
  \textbf{Magic} &
  0.7017 &
  \textbf{Lakers} &
  0.6330 &
  \textbf{Pacers} &
  0.5711 &
  \textbf{Cavaliers} &
  0.7377 \\
\textbf{Rockets} &
  0.2473 &
  \textbf{Spurs} &
  0.5872 &
  \textbf{Mavericks} &
  0.5353 &
  \textbf{Rockets} &
  0.5179 &
  \textbf{Rockets} &
  0.4701 \\
\textbf{Heat} &
  0.1795 &
  \textbf{Mavericks} &
  0.5297 &
  \textbf{Celtics} &
  0.4928 &
  \textbf{Warriors} &
  0.4497 &
  \textbf{Mavericks} &
  0.4178 \\
\textbf{Cavaliers} &
  0.1540 &
  \textbf{Jazz} &
  0.4169 &
  \textbf{Magic} &
  0.4406 &
  \textbf{Nuggets} &
  0.4091 &
  \textbf{Trail Blazers} &
  0.3920 \\
\textbf{Nuggets} &
  0.1502 &
  \textbf{Suns} &
  0.4024 &
  \textbf{Nuggets} &
  0.3650 &
  \textbf{Bulls} &
  0.3619 &
  \textbf{Heat} &
  0.1986 \\
\textbf{Nets} &
  0.0460 &
  \textbf{Nuggets} &
  0.3997 &
  \textbf{Trail Blazers} &
  0.0628 &
  \textbf{Knicks} &
  0.2510 &
  \textbf{Bulls} &
  0.1531 \\
\textbf{Bulls} &
  0.0219 &
  \textbf{Trail Blazers} &
  0.3177 &
  \textbf{Rockets} &
  0.0303 &
  \textbf{Mavericks} &
  0.2371 &
  \textbf{Pacers} &
  0.1095 \\
\textbf{Kings} &
  -0.0057 &
  \textbf{Rockets} &
  0.2942 &
  \textbf{76ers} &
  0.0296 &
  \textbf{Trail Blazers} &
  0.1987 &
  \textbf{Jazz} &
  0.0933 \\
\textbf{Lakers} &
  -0.0220 &
  \textbf{Hornets} &
  0.2866 &
  \textbf{Pacers} &
  0.0008 &
  \textbf{Nets} &
  0.1871 &
  \textbf{Celtics} &
  0.0752 \\
\textbf{Clippers} &
  -0.0649 &
  \textbf{Bulls} &
  -0.0684 &
  \textbf{Knicks} &
  -0.0209 &
  \textbf{Lakers} &
  -0.1328 &
  \textbf{Pistons} &
  0.0251 \\
\textbf{Jazz} &
  -0.0926 &
  \textbf{Pistons} &
  -0.2101 &
  \textbf{Suns} &
  -0.0213 &
  \textbf{Timberwolves} &
  -0.1340 &
  \textbf{Hornets} &
  0.0007 \\
\textbf{Pacers} &
  -0.1213 &
  \textbf{Heat} &
  -0.2717 &
  \textbf{Jazz} &
  -0.0385 &
  \textbf{Suns} &
  -0.2172 &
  \textbf{Bucks} &
  -0.0649 \\
\textbf{Timberwolves} &
  -0.1568 &
  \textbf{Warriors} &
  -0.3373 &
  \textbf{Clippers} &
  -0.0771 &
  \textbf{Jazz} &
  -0.2758 &
  \textbf{Kings} &
  -0.3145 \\
\textbf{Warriors} &
  -0.1936 &
  \textbf{76ers} &
  -0.3645 &
  \textbf{Hornets} &
  -0.3269 &
  \textbf{Celtics} &
  -0.4300 &
  \textbf{Suns} &
  -0.3763 \\
\textbf{Magic} &
  -0.2521 &
  \textbf{Pacers} &
  -0.3706 &
  \textbf{Warriors} &
  -0.3615 &
  \textbf{Kings} &
  -0.5275 &
  \textbf{Magic} &
  -0.3995 \\
\textbf{76ers} &
  -0.2864 &
  \textbf{Bucks} &
  -0.5188 &
  \textbf{Bucks} &
  -0.3695 &
  \textbf{Hornets} &
  -0.5792 &
  \textbf{Nuggets} &
  -0.4144 \\
\textbf{Hornets} &
  -0.4496 &
  \textbf{Kings} &
  -0.7237 &
  \textbf{Pistons} &
  -0.6596 &
  \textbf{Pistons} &
  -0.6031 &
  \textbf{Nets} &
  -0.5219 \\
\textbf{Celtics} &
  -0.4546 &
  \textbf{Knicks} &
  -0.8034 &
  \textbf{Kings} &
  -0.6985 &
  \textbf{Bucks} &
  -0.6512 &
  \textbf{Knicks} &
  -0.9580 \\
\textbf{Bucks} &
  -0.5658 &
  \textbf{Nets} &
  -0.9115 &
  \textbf{Timberwolves} &
  -0.7597 &
  \textbf{Cavaliers} &
  -0.7121 &
  \textbf{Timberwolves} &
  -1.0322 \\
\textbf{Knicks} &
  -0.7506 &
  \textbf{Timberwolves} &
  -1.0381 &
  \textbf{Nets} &
  -0.9547 &
  \textbf{76ers} &
  -0.7279 &
  \textbf{Lakers} &
  -1.0736 \\
\textbf{Trail Blazers} &
  -0.8684 &
  \textbf{Clippers} &
  -1.1249 &
  \textbf{Cavaliers} &
  -1.0610 &
  \textbf{Magic} &
  -0.9004 &
  \textbf{76ers} &
  -1.4612 \\ \bottomrule
\end{tabular}
}
\caption{Fitted scores for 24 selected teams in the seasons S1980-S2016 of NBA. It is divided into 10 time periods based on the 9 estimated change points.  Within each period, the teams are ranked based on their fitted scores.}
%\label{NBA_table2}
\label{NBA_tableall}
\end{table}

Our methods detect 9 different change points, which divide the whole history into 10 time periods.  The NBA facts can explain these 9 change points. To be more specific, the first and second periods, S1980-S1985 and S1986-S1989, represent the times that Larry Bird in the {\it Celtics} and Michael Johnson in the {\it Lakers} ruled this era together. In the third period, S1990-S1994, Michael Jordan in the {\it Bulls} won the triple crown. However, in S1994, Michael Jordan retired for the first time, and the {\it Rockets} won 2 champions. In S1995, Michael Jordan came back, and the {\it Bulls} achieved another triple crown. In S1998, Michael Jordan retired again, and Shaq and Kobe helped the {\it Lakers} dominate the period S1998-S2003. In S2004-S2006, the ``Big 3'' in {\it Spurs} emerged. In S2007-S2009, Kobe helped the {\it Lakers} win another 2 champions. In S2010-S2011, Derk in the {\it Mavericks} defeated the ``Big 3'' in the {\it Heat}, but in S2011-S2013, the {\it Heat} dominated the scene. Lastly, in S2014-S2015, Stephen Curry helped the {\it Warriors} take the lead.

When comparing our results with those in \cite{li2022detecting}, our method detected some important events like Michael Jordan retiring for the first time in S1994 and the {\it Mavericks} defeated {\it Heat} in S2011, while the DPLR method \cite{li2022detecting} failed to detect them. Also, the MDL method has, in terms of minus log-likelihood value, a smaller loss: the loss for DPLR is approximately $8880$ while the MDL's loss is approximately $8000$, which is a $10\%$ decrease.

\section{Concluding Remarks}
\label{section7}
This paper addresses the challenge of detecting change points within sequential pairwise comparison data, incorporating covariate information for enhanced accuracy.  At its core, this paper introduced the piecewise stationary covariate-assisted ranking estimation (PS-CARE) model, an innovative extension of the CARE model designed to handle these data complexities.

We developed a comprehensive methodology for accurately estimating the PS-CARE model's unknown parameters, including both the number and precise locations of change points. Central to our approach is the application of the minimum description length (MDL) principle, which facilitated the derivation of an objective criterion for parameter estimation. It has been shown the MDL estimates are statistically consistent.

The practical optimization of the MDL criterion was achieved using the PELT algorithm. Our extensive simulation experiments underscored the excellent empirical performance of our proposed methodology.  When applied to an NBA dataset, our methodology not only identified meaningful results but also correlated these findings with significant historical events within the dataset's timeline, showcasing the practical relevance of our approach.

In conclusion, this paper contributes significantly to the field of dynamic ranking systems by presenting the PS-CARE model as a powerful tool for change point detection in sequential pairwise comparison data, especially when covariate information is available.  The demonstrated success of the PS-CARE model, with its proven methodological rigor and empirical validation, paves the way for future research and offers valuable insights for practitioners and researchers alike.

\begin{appendices}
\setcounter{equation}{0}
\setcounter{table}{0}

\section{Proof and Technical Details}
\renewcommand{\thetable}{\thesection.\arabic{table}}
\renewcommand{\theequation}{\thesection.\arabic{equation}}

\label{app:proofs}

This appendix presents technical details and the proof for Theorem~\ref{Theorem1}.  We will first define some notations and introduce several lemmas.

\subsection{Lemmas}
The $k$-th segment of $\{Y_t\}$ is modeled by a stationary time series $\boldsymbol{x_k}=\{x_{t,k}\}_{t\in\mathbb{Z}}$ such that 
$y_t = x_{t-\tau_{k-1},k}$ for all $\tau_{k-1}+1\leq t \leq \tau_k$
%\begin{equation*}
%    y_t = x_{t-\tau_{k-1},k}, \quad \forall \tau_{k-1}+1\leq t \leq \tau_k
%\end{equation*}
and $T_k = \tau_{k} - \tau_{k-1}$ for $k=1,\cdots,K+1$. 

Define $l_i(\boldsymbol{\xi}_T^{(j)}; x_{i,k},\boldsymbol{z}^*_t|x_{s,k}, s<i)$ as the conditional log-likelihood function at time $i$ for observations in the $k$-th segment, given all the past observations. And the conditional log-likelihood of $k$-th segment, $\boldsymbol{x_k}=\{x_{t,k}, t=1,2,\cdots,T_j\}$ given all the past observations is 
\begin{equation*}
    L_T^{(k)}(\boldsymbol{\xi}_T^{(k)}; \boldsymbol{x_k}) = \sum_{i=1}^{T_k}l_i(\boldsymbol{\xi}_T^{(k)}; x_{i,k},\boldsymbol{z}^*_t|x_{s,k}, s<i).
\end{equation*}

However, it is impossible to observe all the past observations $\{x_{t,k}\}_{t<0}$ in practice. Denote $\boldsymbol{y_{i,k}}=(y_1,\cdots, y_{\tau_{k-1}+i-1})$ as the observed past in practice. Thus, the observed likelihood for the $k$-th segment is given by 
\begin{equation*}
     \Tilde{L}_T^{(k)}(\boldsymbol{\xi}_T^{(k)}; \boldsymbol{x_k}) = \sum_{i=1}^{T_k}l_i(\boldsymbol{\xi}_T^{(k)}; x_{i,k},\boldsymbol{z}^*_t|\boldsymbol{y_{i,k}}).
\end{equation*}

Now, for estimating the parameters of (\ref{MDL_estimate}), consider the situation where only a portion of the data in the $k$-th segment is chosen to perform the parameter estimation. Define $\lambda_l,\lambda_u\in[0,1]$  and $\lambda_l<\lambda_u, \lambda_u-\lambda_l>\epsilon_\lambda$, where $\epsilon_\lambda$ is defined in (\ref{lambda_assump}). To simplify the notation, we write 
\begin{equation*}
    \sup_{\lambda_l,\lambda_u}=\sup_{\lambda_l,\lambda_u\in[0,1],\lambda_u-\lambda_l>\epsilon_\lambda}.
\end{equation*}

Next, define the true and observed log-likelihood function based on a portion of the data in the $k$-th segment as follows:
\begin{align}
\label{Real_part_like}
     L_T^{(k)}(\boldsymbol{\xi}_T^{(k)},\lambda_l,\lambda_u; \boldsymbol{x_k}) &= \sum_{i=[T_k\lambda_l]+1}^{T_k\lambda_u}l_i(\boldsymbol{\xi}_T^{(k)}; x_{i,k},\boldsymbol{z}^*_t|x_{s,k}, s<i), \\
     \label{obs_part_like}
    \Tilde{L}_T^{(k)}(\boldsymbol{\xi}_T^{(k)},\lambda_l,\lambda_u; \boldsymbol{x_k}) &= \sum_{i=[T_k\lambda_l]+1}^{T_k\lambda_u}l_i(\boldsymbol{\xi}_T^{(k)}; x_{i,k},\boldsymbol{z}^*_t|\boldsymbol{y_{i,k}}).
\end{align}

In practice, we can only use (\ref{obs_part_like}) instead of (\ref{Real_part_like}); thus, our first lemma controls the quality of the approximation.

\begin{lemma}
\label{lemma_1}
For any $k=1,\cdots,K+1$, the first and second derivatives $L_n^{'(k)}$, $\Tilde{L}_n^{'(k)}$ and $L_n^{''(k)}$, $\Tilde{L}_n^{''(k)}$ with respect to $\boldsymbol{\xi}_T^{(k)}$ of the function defined in (\ref{obs_part_like}) and (\ref{Real_part_like}) satisfy
\begin{align}
\label{1emma_1_1}
    \sup _{\lambda_l, \lambda_u} \sup _{\boldsymbol{\xi}_T^{(k)} \in \Theta_n}\left|\frac{1}{T} L_T^{(k)}\left(\boldsymbol{\xi}_T^{(k)}, \lambda_l, \lambda_u, \mathbf{x}_{\mathbf{k}}\right)-\frac{1}{T} \tilde{L}_T^{(k)}\left(\boldsymbol{\xi}_T^{(k)}, \lambda_l, \lambda_u, \mathbf{x}_{\mathbf{k}}\right)\right|&=o\left(1\right), \\
    \label{lemma_1_2}
    \sup _{\lambda_l, \lambda_u} \sup _{\boldsymbol{\xi}_T^{(k)} \in \Theta_n}\left|\frac{1}{T} L_T^{'(k)}\left(\boldsymbol{\xi}_T^{(k)}, \lambda_l, \lambda_u, \mathbf{x}_{\mathbf{k}}\right)-\frac{1}{T} \tilde{L}_T^{'(k)}\left(\boldsymbol{\xi}_T^{(k)}, \lambda_l, \lambda_u, \mathbf{x}_{\mathbf{k}}\right)\right|&=o\left(1\right), \\
    \label{lemma_1_3}
    \sup _{\lambda_l, \lambda_u} \sup _{\boldsymbol{\xi}_T^{(k)} \in \Theta_n}\left|\frac{1}{T} L_T^{''(k)}\left(\boldsymbol{\xi}_T^{(k)}, \lambda_l, \lambda_u, \mathbf{x}_{\mathbf{k}}\right)-\frac{1}{T} \tilde{L}_T^{''(k)}\left(\boldsymbol{\xi}_T^{(k)}, \lambda_l, \lambda_u, \mathbf{x}_{\mathbf{k}}\right)\right|&=o\left(1\right).
\end{align}
\end{lemma}

\begin{proof}
We shall only prove~(\ref{1emma_1_1}), while~(\ref{lemma_1_2}) and (\ref{lemma_1_3}) can be proved using similar arguments. For the PS-CARE model defined in (\ref{p_ij_CARE}), we have 
\begin{align*}
    L_T^{(k)}\left(\boldsymbol{\xi}_T^{(k)}, \lambda_l, \lambda_u, \mathbf{x}_{\mathbf{k}}\right) &= \sum_{i=[T_k\lambda_l]+1}^{T_k\lambda_u}l_i(\boldsymbol{\xi}_T^{(k)}; x_{i,k},\boldsymbol{z}^*_t|x_{s,k}, s<i)  \nonumber \\ &= \sum_{i=[T_k\lambda_l]+1}^{T_k\lambda_u}\left(x_{i,k}\boldsymbol{z}_i^{*T}\boldsymbol{\hat{\xi}}_{T}^{(k,1)} - \log(1+\exp(\boldsymbol{z}_i^{*T}\boldsymbol{\hat{\xi}}_{T}^{(k,1)}))\right)
\end{align*}
and
\begin{align*}
    \Tilde{L}_T^{(k)}\left(\boldsymbol{\xi}_T^{(k)}, \lambda_l, \lambda_u, \mathbf{x}_{\mathbf{k}}\right) &= \sum_{i=[T_k\lambda_l]+1}^{T_k\lambda_u}l_i(\boldsymbol{\xi}_T^{(k)}; x_{i,k},\boldsymbol{z}^*_t|\boldsymbol{y_{i,k}})  \nonumber \\ &= \sum_{i=[T_k\lambda_l]+1}^{T_k\lambda_u}\left(x_{i,k}\boldsymbol{z}_i^{*T}\boldsymbol{\hat{\xi}}_{T}^{(k,2)} - \log(1+\exp(\boldsymbol{z}_i^{*T}\boldsymbol{\hat{\xi}}_{T}^{(k,2)}))\right),
\end{align*}
where $\boldsymbol{\hat{\xi}}_{T}^{(k,1)}$ and $\boldsymbol{\hat{\xi}}_{T}^{(k,2)}$ are the maximum likelihood estimators based on true past and observed past in the $j-th$ segment respectively.

So, we have 
\begin{equation}
\begin{aligned}
\label{proof_lemma1}
   & \sup _{\lambda_l, \lambda_u} \sup _{\boldsymbol{\xi}_T^{(k)} \in \Theta_n}\left|\frac{1}{T} L_T^{(k)}\left(\boldsymbol{\xi}_T^{(k)}, \lambda_l, \lambda_u, \mathbf{x}_{\mathbf{k}}\right)-\frac{1}{T} \tilde{L}_T^{(k)}\left(\boldsymbol{\xi}_T^{(k)}, \lambda_l, \lambda_u, \mathbf{x}_{\mathbf{k}}\right)\right| \\
    &\leq \sup _{\lambda_l, \lambda_u} \frac{1}{T}\sum_{i=[T_k\lambda_l]+1}^{T_j\lambda_u} \left(x_{i,k}\left| \boldsymbol{z}_i^{*T}\boldsymbol{\hat{\xi}}_{T}^{(k,1)} - \boldsymbol{z}_i^{*T}\boldsymbol{\hat{\xi}}_{T}^{(k,2)} \right| + \left| \log\left(1+\frac{\exp(\boldsymbol{z}_i^{*T}\boldsymbol{\hat{\xi}}_{T}^{(k,2)})) - \exp(\boldsymbol{z}_i^{*T}\boldsymbol{\hat{\xi}}_{T}^{(k,1)}))}{1+\exp(\boldsymbol{z}_i^{*T}\hat{{\boldsymbol{\xi}}}_{k,1}))}\right) \right|\right) \\
    &\leq \sup _{\lambda_l, \lambda_u} \frac{1}{T}\sum_{i=[T_k\lambda_l]+1}^{T_k\lambda_u} \left(x_{i,k}\| \boldsymbol{z}_i^{*T}\boldsymbol{\hat{\xi}}_{T}^{(k,1)} - \boldsymbol{z}_i^{*T}\boldsymbol{\hat{\xi}}_{T}^{(k,2)} \|_{\infty} +\log\left(1+ \|  \frac{\exp(\boldsymbol{z}_i^{*T}\boldsymbol{\hat{\xi}}_{T}^{(k,2)})) - \exp(\boldsymbol{z}_i^{*T}\boldsymbol{\hat{\xi}}_{T}^{(k,1)}))}{1+\exp(\boldsymbol{z}_i^{*T}\boldsymbol{\hat{\xi}}_{T}^{(k,1)}))} \|_{\infty}\right)\right).
\end{aligned}
\end{equation}

Assume $\Tilde{\boldsymbol{\beta}}_T^*$ be the true parameter vector of $j$-th segment. According to Theorem 3.1 in \cite{fan2022uncertainty}, as long as Assumptions~\ref{assumption_1} to~\ref{assumption_3} are satisfied, we have 
\begin{align*}
    \|\Tilde{\boldsymbol{Z}}\boldsymbol{\hat{\xi}}_{T}^{(k,i)}-\Tilde{\boldsymbol{Z}}\Tilde{\boldsymbol{\beta}}_T^*\|_{\infty} = O(L_{min,k}^{-1/2}), \text{ for } i = 1, 2, \\
    \frac{\left\|e^{\widetilde{\boldsymbol{Z}} \boldsymbol{\hat{\xi}}_{T}^{(k,i)}}-e^{\widetilde{\boldsymbol{X}} \widetilde{\boldsymbol{\beta}}_T^*}\right\|_{\infty}}{\left\|e^{\widetilde{\boldsymbol{Z}} \widetilde{\boldsymbol{\beta}}_T^*}\right\|_{\infty}} = O(L_{min,k}^{-1/2}),  \text{ for } i = 1, 2.
\end{align*}
And we have 
\begin{align*}
    \| \boldsymbol{z}_i^{*T}\boldsymbol{\hat{\xi}}_{T}^{(k,1)} - \boldsymbol{z}_i^{*T}\boldsymbol{\hat{\xi}}_{T}^{(k,2)} \|_{\infty} & \leq \| \boldsymbol{z}_i^{*T}\boldsymbol{\hat{\xi}}_{T}^{(k,1)} - \boldsymbol{z}_i^{*T}\hat{{\boldsymbol{\xi}}}_T^* \|_{\infty} + \| \boldsymbol{z}_i^{*T}\boldsymbol{\hat{\xi}}_{T}^{(k,2)} - \boldsymbol{z}_i^{*T}\hat{{\boldsymbol{\xi}}}_T^* \|_{\infty}, \\ 
    \left\|  \frac{\exp(\boldsymbol{z}_i^{*T}\boldsymbol{\hat{\xi}}_{T}^{(k,2)})) - \exp(\boldsymbol{z}_i^{*T}\boldsymbol{\hat{\xi}}_{T}^{(k,1)}))}{1+\exp(\boldsymbol{z}_i^{*T}\boldsymbol{\hat{\xi}}_{T}^{(k,1)}))} \right\|_{\infty} & \leq \exp(\|\boldsymbol{z}_i^{*T}\boldsymbol{\hat{\xi}}_{T}^{(k,1)} - \boldsymbol{z}_i^{*T}\hat{{\boldsymbol{\xi}}}_T^* \|_{\infty}) + \exp(\|\boldsymbol{z}_i^{*T}\boldsymbol{\hat{\xi}}_{T}^{(k,2)} - \boldsymbol{z}_i^{*T}\hat{{\boldsymbol{\xi}}}_T^* \|_{\infty}).
\end{align*}
Thus (\ref{proof_lemma1}) $= o(1)$ is satisfied. 
\end{proof}

\begin{lemma}
For $k = 1, \cdots, K+1$, there exists an $\epsilon>0$ such that 
\label{lemma_2}
\begin{align*}
    \sup _{\boldsymbol{\xi}_T^{(k)} \in \Theta_n} E\left|l_k\left(\boldsymbol{\xi}_T^{(k)} ; x_{1, k} \mid x_{l, k}, l<1\right)\right|^{\epsilon}&<\infty, \\ 
    \sup _{\boldsymbol{\xi}_T^{(k)} \in \Theta_n} E\left|l_k^{'}\left(\boldsymbol{\xi}_T^{(k)} ; x_{1, k} \mid x_{l, k}, l<1\right)\right|^{\epsilon}&<\infty, \\
    \sup _{\boldsymbol{\xi}_T^{(k)} \in \Theta_n} E\left|l_k^{''}\left(\boldsymbol{\xi}_T^{(k)} ; x_{1, k} \mid x_{l, k}, l<1\right)\right|&<\infty.
\end{align*}
\end{lemma}

\begin{proof}
Since $l_k\left(\boldsymbol{\xi}_T^{(k)} ; x_{1, k} \mid x_{l, k}, l<1\right) = x_{1,k}\boldsymbol{z}_i^{*T}\boldsymbol{\hat{\xi}}_T^{(k)} - \log(1+\exp(\boldsymbol{z}_i^{*T}\boldsymbol{\hat{\xi}}_T^{(k)}))$, which is the probability of one item beating another item, and hence within $(0,1)$. Thus, Lemma \ref{lemma_2} is proved.
%\tlee{where is items $g$ and $h$?} {\color{blue}{Here the $\boldsymbol{z}_i^*$ represent at the time point i, 2 items are compared, see paragraph above equation (\ref{log_like_one_observation}) for more detailed explanation, here I don't assume specific item $g,h$, maybe we can change it to the probability of winning.}}
\end{proof}

\begin{lemma}
\label{lemma_3}
For each $k=1,\cdots,K+1$, 
\begin{align}
\label{lemma_3_1}
    \sup _{\boldsymbol{\xi}_T^{(k)} \in \Theta_n}\left|\frac{1}{T} L_T^{(k)}\left( \boldsymbol{\xi}_T^{(k)} ; \mathbf{x}_k\right)-L_k\left(\boldsymbol{\xi}_T^{(k)}\right)\right| \stackrel{a . s}{\longrightarrow} 0, \\
    \label{lemma_3_2}
    \sup _{\boldsymbol{\xi}_T^{(k)} \in \Theta_n}\left|\frac{1}{T} L_T^{'(k)}\left( \boldsymbol{\xi}_T^{(k)}; \mathbf{x}_k\right)-L_j^{'}\left(\boldsymbol{\xi}_T^{(k)}\right)\right| \stackrel{a . s}{\longrightarrow} 0, \\
    \label{lemma_3_3}
    \sup _{\boldsymbol{\xi}_T^{(k)} \in \Theta_n}\left|\frac{1}{T} L_T^{''(k)}\left( \boldsymbol{\xi}_T^{(k)} ; \mathbf{x}_k\right)-L_k^{''}\left(\boldsymbol{\xi}_T^{(k)}\right)\right| \stackrel{a . s}{\longrightarrow} 0, 
\end{align}
where 
\begin{align*}
    L_k\left(\boldsymbol{\xi}_T^{(k)}\right) = E\left(l_k\left(\boldsymbol{\xi}_T^{(k)} ; x_{1, k} \mid x_{l, k}, l<1\right)\right), \\
    L_k^{'}\left(\boldsymbol{\xi}_T^{(k)}\right) = E\left(l_k^{'}\left(\boldsymbol{\xi}_T^{(k)} ; x_{1, k} \mid x_{l, k}, l<1\right)\right), \\
    L_k^{''}\left(\boldsymbol{\xi}_T^{(k)}\right) = E\left(l_k^{''}\left(\boldsymbol{\xi}_T^{(k)} ; x_{1, k} \mid x_{l, k}, l<1\right)\right).
\end{align*}
\end{lemma}

\begin{proof}
Here we only prove (\ref{lemma_3_3}), as (\ref{lemma_3_1}) and (\ref{lemma_3_2}) can be proved using similar arguments. Since $\{\boldsymbol{x_k}\}$ is a stationary ergodic process, we only need to prove 
\begin{equation*}
    \frac{1}{T}\sum_{i=1}^{[T\lambda_k]}l_i(\boldsymbol{\xi}_T^{(k)}; x_{i,k},\boldsymbol{z}^*_t|x_{s,k}, s<i) \stackrel{a.s}{\longrightarrow} \lambda_kE\left(l_k\left(\boldsymbol{\xi}_T^{(k)} ; x_{1, k} \mid x_{l, k}, l<1\right)\right).
\end{equation*}

This can be proved by the ergodic theorem. Let $\mathbb{Q}_{[0,1]}$ be the set of rational numbers in $[0,1]$. For $\lambda_j\in\mathbb{Q}_{[0,1]}$,
\begin{equation}
\label{lemma3_pf1}
    \frac{1}{T}\sum_{i=1}^{[T\lambda_k]}l_i(\boldsymbol{\xi}_T^{(k)}; x_{i,k},\boldsymbol{z}^*_t|x_{s,k}, s<i) \stackrel{a.s}{\longrightarrow} \lambda_kE\left(l_k\left(\boldsymbol{\xi}_T^{(k)} ; x_{1, k} \mid x_{l, k}, l<1\right)\right).
\end{equation}
If $B_\lambda$ is the set of $\omega$'s for which (\ref{lemma3_pf1}) holds, then set $B=\cap_{\lambda_k\in\mathbb{Q}_{[0,1]}}B_\lambda$ and $P(B)=1$. Moreover, for $\omega\in B$ and any $s\in[0,1]$, choose $\lambda_1,\lambda_2\in\mathbb{Q}_{[0,1]}$ such that $\lambda_1\leq \lambda_k \leq\lambda_2$, hence
\begin{equation*}
\begin{aligned}
    &\left|\frac{1}{T}\sum_{i=1}^{[T\lambda_k]}l_i(\boldsymbol{\xi}_T^{(k)}; x_{i,k},\boldsymbol{z}^*_t|x_{s,k}, s<i) - \frac{1}{T}\sum_{i=1}^{[T\lambda_1]}l_i(\boldsymbol{\xi}_T^{(k)}; x_{i,k},\boldsymbol{z}^*_t|x_{s,k}, s<i) \right| \\ 
    &\leq \frac{1}{T}\sum_{i=[T\lambda_1]}^{[T\lambda_2]}\left|l_i(\boldsymbol{\xi}_T^{(k)}; x_{i,k},\boldsymbol{z}^*_t|x_{s,k}, s<i)\right| \longrightarrow (\lambda_2-\lambda_1)E\left|l_i(\boldsymbol{\xi}_T^{(k)}; x_{1,k},\boldsymbol{z}^*_t|x_{s,k}, s<i)\right|.
\end{aligned}
\end{equation*}

By making $\lambda_2-\lambda_1$ arbitrarily small, it follows from the ergodic theorem that 
\begin{equation*}
    \frac{1}{T}\sum_{i=1}^{[T\lambda_k]}l_i(\boldsymbol{\xi}_T^{(k)}; x_{i,k},\boldsymbol{z}^*_t|x_{s,k}, s<i) \longrightarrow \lambda_kE(l_i(\boldsymbol{\xi}_T^{(k)}; x_{1,k},\boldsymbol{z}^*_t|x_{s,k}, s<i)).
\end{equation*}

To establish convergence on $D[0,1]$, it is suffice to show that for any $\omega\in B$, we have
\begin{equation*}
    \frac{1}{T}\sum_{i=1}^{[T\lambda_k]}l_i(\boldsymbol{\xi}_T^{(k)}; x_{i,k},\boldsymbol{z}^*_t|x_{s,k}, s<i) \longrightarrow \lambda_kE(l_i(\boldsymbol{\xi}_T^{(k)}; x_{1,k},\boldsymbol{z}^*_t|x_{s,k}, s<i)) \text{  uniformly on [0,1].   }
\end{equation*}

Since $\epsilon>0$, we can choose $\lambda_1,\lambda_2,\cdots,\lambda_{K}\in\mathbb{Q}_{[0,1]}$ such that $0=\lambda_0<\lambda_1<\cdots<\lambda_{K+1}=1$, with $\lambda_{i} - \lambda_{i-1}<\epsilon$. Then for any $\lambda_k\in[0,1]$, $\lambda_{i-1}<\lambda_{k}\leq \lambda_{i}$ and 
\begin{equation*}
\begin{aligned}
    &\left|\frac{1}{T}\sum_{i=1}^{[T\lambda_k]}l_i(\boldsymbol{\xi}_T^{(k)}; x_{i,k},\boldsymbol{z}^*_t|x_{s,k}, s<i) - \lambda_kE(l_i(\boldsymbol{\xi}_T^{(k)}; x_{1,k},\boldsymbol{z}^*_t|x_{s,k}, s<i) )\right| \\
    \leq & \left| \frac{1}{T}\sum_{i=1}^{[T\lambda_k]}l_i(\boldsymbol{\xi}_T^{(k)}; x_{i,k},\boldsymbol{z}^*_t|x_{s,k}, s<i) - \frac{1}{T}\sum_{i=1}^{[T\lambda_{i-1}]}l_i(\boldsymbol{\xi}_T^{(k)}; x_{i,k},\boldsymbol{z}^*_t|x_{s,k}, s<i)\right|  \\
    & + \left| \frac{1}{T}\sum_{i=1}^{[T\lambda_{i-1}]}l_i(\boldsymbol{\xi}_T^{(k)}; x_{i,k},\boldsymbol{z}^*_t|x_{s,k}, s<i) - \lambda_{i-1}E(l_i(\boldsymbol{\xi}_T^{(k)}; x_{1,k},\boldsymbol{z}^*_t|x_{s,k}, s<i))\right| \\
    &+ \left|\lambda_{i-1}E(l_i(\boldsymbol{\xi}_T^{(k)}; x_{1,k},\boldsymbol{z}^*_t|x_{s,k}, s<i))-\lambda_kE(l_i(\boldsymbol{\xi}_T^{(k)}; x_{1,k},\boldsymbol{z}^*_t|x_{s,k}, s<i))\right|,
\end{aligned}
\end{equation*}
where the first term is bounded by 
\begin{equation*}
\begin{aligned}
    \frac{1}{T}\sum_{i=[T\lambda_{i-1}]}^{[T\lambda_{i}]}|l_i(\boldsymbol{\xi}_T^{(k)}; x_{i,k},\boldsymbol{z}^*_t|x_{s,k}, s<i)| &\longrightarrow (\lambda_{i} - \lambda_{i-1}) E\left|l_i(\boldsymbol{\xi}_T^{(k)}; x_{1,k},\boldsymbol{z}^*_t|x_{s,k}, s<i)\right| \\
    &< \epsilon E\left|l_i(\boldsymbol{\xi}_T^{(k)}; x_{1,k},\boldsymbol{z}^*_t|x_{s,k}, s<i)\right|.
\end{aligned}
\end{equation*}

Let $T$ be large enough so that this term is less than $\epsilon E\left|l_i(\boldsymbol{\xi}_T^{(k)}; x_{1,k},\boldsymbol{z}^*_t|x_{s,k}, s<i)\right|$ for $i=1,\cdots,K$. So it follows that
\begin{equation*}
\begin{aligned}
    &\left|\frac{1}{T}\sum_{i=1}^{[T\lambda_k]}l_i(\boldsymbol{\xi}_T^{(k)}; x_{i,k},\boldsymbol{z}^*_t|x_{s,k}, s<i) - \lambda_kE(l_i(\boldsymbol{\xi}_T^{(k)}; x_{1,k},\boldsymbol{z}^*_t|x_{s,k}, s<i) )\right|  \\
    &< \epsilon E|l_i(\boldsymbol{\xi}_T^{(k)}; x_{1,k},\boldsymbol{z}^*_t|x_{s,k}, s<i)| + \epsilon + \epsilon E|l_i(\boldsymbol{\xi}_T^{(k)}; x_{1,k},\boldsymbol{z}^*_t|x_{s,k}, s<i)|.
\end{aligned}
\end{equation*}
Since $\epsilon$ can be arbitrarily small, (\ref{lemma_3_1}) is proved, and (\ref{lemma_3_2}) and (\ref{lemma_3_3}) can be proved in a similar manner.
\end{proof}

\begin{lemma}
\label{lemma_4}
For the PS-CARE model defined above, we have 
\begin{align}
\label{cor_1_1}
    \sup _{\lambda_l, \lambda_u} \sup _{\boldsymbol{\xi}_T^{(k)} \in \Theta_n}\left|\frac{1}{T} \tilde{L}_T^{(k)}\left(\boldsymbol{\xi}_T^{(k)}, \lambda_l, \lambda_u ; \mathbf{x}_k\right)-\left(\lambda_l-\lambda_u\right) L_k\left(\boldsymbol{\xi}_T^{(k)}\right)\right| \stackrel{a . s .}{\longrightarrow} 0, \\
    \label{cor_1_2}
    \sup _{\lambda_l, \lambda_u} \sup _{\boldsymbol{\xi}_T^{(k)} \in \Theta_n}\left|\frac{1}{T} \tilde{L}_T^{'(k)}\left(\boldsymbol{\xi}_T^{(k)}, \lambda_l, \lambda_u ; \mathbf{x}_k\right)-\left(\lambda_l-\lambda_u\right) L_k^{'}\left(\boldsymbol{\xi}_T^{(k)}\right)\right| \stackrel{a . s .}{\longrightarrow} 0, \\
    \label{cor_1_3}
    \sup _{\lambda_l, \lambda_u} \sup _{\boldsymbol{\xi}_T^{(k)} \in \Theta_n}\left|\frac{1}{T} \tilde{L}_T^{''(k)}\left(\boldsymbol{\xi}_T^{(k)}, \lambda_l, \lambda_u ; \mathbf{x}_k\right)-\left(\lambda_l-\lambda_u\right) L_k^{''}\left(\boldsymbol{\xi}_T^{(k)}\right)\right| \stackrel{a . s .}{\longrightarrow} 0.
\end{align}
\end{lemma}

\begin{proof}
Here only prove (\ref{cor_1_1}), as (\ref{cor_1_2}) and (\ref{cor_1_3}) can be proved in a similar manner. From Lemma (\ref{lemma_1}), we only need to prove $L_T^{(k)}\left(\boldsymbol{\xi}_T^{(k)}, \lambda_l, \lambda_u ; \mathbf{x}_k\right)$ instead of $\tilde{L}_T^{(k)}\left(\boldsymbol{\xi}_T^{(k)}, \lambda_l, \lambda_u ; \mathbf{x}_k\right)$. Let $\mathbb{Q}_{[0,1]}$ be a set of rational numbers in $[0,1]$. Then $\forall r_1, r_2\in\mathbb{Q}_{[0,1]}$ with $r_1<r_2$, by (\ref{lemma_3}), we have 
\begin{equation}
\begin{aligned}
\label{cor_1_proof}
    &\sup _{\boldsymbol{\xi}_T^{(k)} \in \Theta_n}\left|\frac{1}{T}L_T^{(k)}\left(\boldsymbol{\xi}_T^{(k)}, \lambda_l, \lambda_u ; \mathbf{x}_k\right) - \left(r_2-r_1\right) L_k\left(\boldsymbol{\xi}_T^{(k)}\right) \right| \\ 
    = & \sup _{\boldsymbol{\xi}_T^{(k)} \in \Theta_n}|r_2\left(\frac{1}{Tr_2}\sum_{i=1}^{[Tr_2]}l_i(\boldsymbol{\xi}_T^{(k)}; x_{i,k},\boldsymbol{z}^*_t|x_{s,k}, s<i) - L_k\left(\boldsymbol{\xi}_T^{(k)}\right)\right) \\
    &- r_1\left(\frac{1}{Tr_1}\sum_{i=1}^{[Tr_1]}l_i(\boldsymbol{\xi}_T^{(k)}; x_{i,k},\boldsymbol{z}^*_t|x_{s,k}, s<i)-- L_k\left(\boldsymbol{\xi}_T^{(k)}\right)\right) | \stackrel{a . s .}{\longrightarrow} 0.
\end{aligned}
\end{equation}

Let $B_{r_1,r_2}$ be the probability one set of $\omega's$ for which (\ref{cor_1_proof}) holds. Set 
\begin{equation*}
    B = \cap _{r_1,r_2\in\mathbb{Q}_{[0,1]}} B_{r_1,r_2}.
\end{equation*}
It is well-known that $P(B) = 1$. Moreover for any $\omega\in B$ and any $\lambda\in[0,1]$, we can choose $r_l, r_u\in\mathbb{Q}_{[0,1]}$ such that $r_l\leq\lambda\leq r_u$. So we have 
\begin{equation*}
\begin{aligned}
    &\sup _{\boldsymbol{\xi}_T^{(k)} \in \Theta_n}\left|\frac{1}{T}\sum_{i=1}^{[T\lambda]}l_i(\boldsymbol{\xi}_T^{(k)}; x_{i,k},\boldsymbol{z}^*_t|x_{s,k}, s<i) - \frac{1}{T}\sum_{i=1}^{[Tr_l]}l_i(\boldsymbol{\xi}_T^{(k)}; x_{i,k},\boldsymbol{z}^*_t|x_{s,k}, s<i) \right| \\
    &\leq \sup _{\boldsymbol{\xi}_T^{(k)} \in \Theta_n} \frac{1}{T}\sum_{i=[Tr_l]+1}^{[Tr_u]}\left|l_i(\boldsymbol{\xi}_T^{(k)}; x_{i,k},\boldsymbol{z}^*_t|x_{s,k}, s<i)\right| \\
    & \longrightarrow (r_u - r_l)  \sup _{\boldsymbol{\xi}_T^{(k)} \in \Theta_n} E\left|l_i(\boldsymbol{\xi}_T^{(k)}; x_{1,k},\boldsymbol{z}^*_t|x_{s,k}, s<i)\right|.
\end{aligned}
\end{equation*}

From Lemma \ref{lemma_2}, we have $\sup _{\boldsymbol{\xi}_T^{(k)} \in \Theta_n} E\left|l_i(\boldsymbol{\xi}_T^{(k)}; x_{i,k},\boldsymbol{z}^*_t|x_{s,k}, s<i)\right|<\infty$. So let $r_u-r_l<\epsilon$ where $\epsilon$ can be arbitrarily small, we have $L_T^{(k)}\left(\boldsymbol{\xi}_T^{(k)}, 0, \lambda ; \mathbf{x}_k\right)\stackrel{a . s .}{\rightarrow} \lambda L_k\left(\boldsymbol{\xi}_T^{(k)}\right) $ uniformly in $\boldsymbol{\xi}_T^{(k)}\in\Theta_n$. With the same idea, we have 
\begin{equation}
    \label{cor_proof_2}
    \sup _{\boldsymbol{\xi}_T^{(k)} \in \Theta_n}\left|\frac{1}{T} L_T^{(k)}\left(\boldsymbol{\xi}_T^{(k)}, \lambda_l, \lambda_u ; \mathbf{x}_k\right)-\left(\lambda_l-\lambda_u\right) L_k\left(\boldsymbol{\xi}_T^{(k)}\right)\right| \stackrel{a . s .}{\longrightarrow} 0
\end{equation}
for any $\lambda_l$ and $\lambda_u$ in $[0,1]$ with $\lambda_l < \lambda_u$. The next step is to show the convergence in (\ref{cor_proof_2}) is uniform in $\lambda_l, \lambda_u$ with $\lambda_u-\lambda_l>\epsilon_\lambda$. For any fixed positive $\epsilon < \epsilon_\lambda$, choose a large $K_1$ such that with $r_0,\cdots,r_{K_1}\in\mathbb{Q}_{[0,1]}$ such that $0=r_0<r_1<\cdots<r_{K_1}=1$ and $\max_{i=1,\cdots,K_1}\leq\epsilon$. Then for any $\lambda_l,\lambda_u\in[0,1]$, we can find $g$ and $h$ such that $g<h$ and $r_{g-1}<\lambda_l<r_g$, $r_{h-1}<\lambda_u<r_{h}$. Thus we have 

\begin{equation*}
\begin{aligned}
&\left|\frac{1}{T} L_T^{(k)}\left(\boldsymbol{\xi}_T^{(k)}, \lambda_l, \lambda_u ; \mathbf{x}_k\right)-\left(\lambda_l-\lambda_u\right) L_k\left(\boldsymbol{\xi}_T^{(k)}\right)\right| \\
\leq & \left|\frac{1}{T} L_T^{(k)}\left(\boldsymbol{\xi}_T^{(k)}, \lambda_l, \lambda_u ; \mathbf{x}_k\right)-\frac{1}{T} L_T^{(k)}\left(\boldsymbol{\xi}_T^{(k)}, r_{g-1}, r_{h} ; \mathbf{x}_k\right)\right| \\
&+ \left|\frac{1}{T} L_T^{(k)}\left(\boldsymbol{\xi}_T^{(k)}, r_{g-1}, r_{h} ; \mathbf{x}_k\right)-\left(r_{h}-r_{g-1}\right) L_k\left(\boldsymbol{\xi}_T^{(k)}\right)\right| \\
&+ \left|\left(r_{g}-r_{h-1}\right) L_k\left(\boldsymbol{\xi}_T^{(k)}\right)-\left(\lambda_l-\lambda_u\right) L_k\left(\boldsymbol{\xi}_T^{(k)}\right)\right|.
\end{aligned}
\end{equation*}

Let $T$ be large enough and the third term is almost surely bounded by 
\begin{equation*}
    \sup _{\boldsymbol{\xi}_T^{(k)} \in \Theta_n} \left|\left(r_{g}-r_{h-1}\right) L_k\left(\boldsymbol{\xi}_T^{(k)}\right)+\left(\lambda_l-\lambda_u\right) L_k\left(\boldsymbol{\xi}_T^{(k)}\right)\right| < 2\epsilon\sup _{\boldsymbol{\xi}_T^{(k)} \in \Theta_n}E\left|l_i(\boldsymbol{\xi}_T^{(k)}; x_{1,k},\boldsymbol{z}^*_t|x_{s,k}, s<i)\right|. 
\end{equation*}
By (\ref{cor_1_proof}), the second term is bounded by $\epsilon$ for sufficiently large $T$. It follows that 
\begin{equation*}
\begin{aligned}
    &\sup _{\lambda_l, \lambda_u} \sup _{\boldsymbol{\xi}_T^{(k)} \in \Theta_n}\left|\frac{1}{T} \tilde{L}_T^{(k)}\left(\boldsymbol{\xi}_T^{(k)}, \lambda_l, \lambda_u ; \mathbf{x}_k\right)-\left(\lambda_l-\lambda_u\right) L_k\left(\boldsymbol{\xi}_T^{(k)}\right)\right| \\ 
    &< 2\epsilon \sup _{\boldsymbol{\xi}_T^{(k)} \in \Theta_n} E\left|l_i(\boldsymbol{\xi}_T^{(k)}; x_{1,k},\boldsymbol{z}^*_t|x_{s,k}, s<i)\right| + \epsilon + 2\epsilon\sup _{\boldsymbol{\xi}_T^{(k)} \in \Theta_n} E\left|l_i(\boldsymbol{\xi}_T^{(k)}; x_{1,k},\boldsymbol{z}^*_t|x_{s,k}, s<i)\right|,
\end{aligned}
\end{equation*}
for sufficiently large $T$, a.s. And since $\epsilon$ can be arbitrarily small, and independent of $\lambda_l,\lambda_u$, thus (\ref{cor_1_1}) is proved.
\end{proof}

\begin{lemma}
\label{lemma_5}
Lemma \ref{lemma_4} also holds if we substitute $\sup_{\lambda_l,\lambda_u}$ by $\sup\limits_{\underline{\lambda_l}, \overline{\lambda_u}}$. Where $\sup \limits_{\underline{\lambda_l}, \overline{\lambda_u}} = \sup \limits_{-h_n<\lambda_l<\lambda_u<1+k_n \atop \lambda_u - \lambda_l >\epsilon_\lambda}$ for any pre-specified sequence $h_n$ and $k_n$ which are converging to $0$ as $n\rightarrow\infty$.
\end{lemma}
\begin{proof}
First define $\grave{\lambda}_l=\max \left(0, \lambda_l\right)$, $\ddot{\lambda}_l=\min \left(0, \lambda_l\right)$, $\grave{\lambda}_u=\min \left(1, \lambda_u\right)$ and $\ddot{\lambda}_u=\max \left(1, \lambda_u\right)$. Then we have 
\begin{equation}
\begin{aligned}
\label{coro_prove_2_1}
&\frac{1}{T_k} L_T^{(k)}\left(\boldsymbol{\xi}_T^{(k)}, \lambda_l, \lambda_u ; \mathbf{x}_k\right)-\left(\lambda_l-\lambda_u\right) L_k\left(\boldsymbol{\xi}_T^{(k)}\right)  \\
= & \frac{1}{T_k} L_T^{(k)}\left(\boldsymbol{\xi}_T^{(k)}, \grave{\lambda}_l, \grave{\lambda}_u ; \mathbf{x}_k\right) - \left(\grave{\lambda}_l-\grave{\lambda}_u\right) L_k\left(\boldsymbol{\xi}_T^{(k)}\right) - (\ddot{\lambda}_u-1-\ddot{\lambda}_l)L_k\left(\boldsymbol{\xi}_T^{(k)}\right) \\
&+\frac{1}{T_k}\sum_{i=T_{k-1}+[T_k(\ddot{\lambda}_l)]+1}^{T_{k-1}}l_i(\boldsymbol{\xi}_T^{(k)}; x_{i,k-1},\boldsymbol{z}^*_t|x_{s,k-1}, s<i) \\
&+ \frac{1}{T_k}\sum_{i=1}^{[T_{k}(\ddot(\lambda)_u-1)]}l_i(\boldsymbol{\xi}_T^{(k)}; x_{i,k+1},\boldsymbol{z}^*_t|x_{s,k+1}, s<i).
\end{aligned}
\end{equation}

Since $0\leq\grave{\lambda}_l<\grave{\lambda}_u\leq 1$, the sum of the first two terms in (\ref{coro_prove_2_1}) converges to $0$ a.s by Lemma \ref{lemma_4}. And for any $\delta>0$, $\max(|\ddot{\lambda}_l,|\ddot{\lambda}_u-1|)<\delta$ for sufficiently large enough $T$. Thus the third term in (\ref{coro_prove_2_1}) is bounded by $2\delta|L_k\left(\boldsymbol{\xi}_T^{(k)}\right)|$. The fourth term is bounded by 
\begin{equation*}
    \frac{1}{T_k}\sum_{i=T_{k-1}-[T_k\delta]}^{T_{k-1}}l_i(\boldsymbol{\xi}_T^{(k)}; x_{i,k-1},\boldsymbol{z}^*_t|x_{s,k-1}, s<i)\stackrel{a.s}\longrightarrow\delta E|l_i(\boldsymbol{\xi}_T^{(k)}; x_{i,k-1},\boldsymbol{z}^*_t|x_{s,k-1}, s<i)|.
\end{equation*}
The last term in (\ref{coro_prove_2_1}) can be bounded similarly. And since $\delta$ can be arbitrarily small, so (\ref{coro_prove_2_1}) converges to $0$ uniformly in $\lambda_l$, $\lambda_u$.
\end{proof}

\begin{lemma} 
\label{lemma_6}
Let $\Tilde{\boldsymbol{\beta}}_k^0$ be the true model parameter. Define 
\begin{align*}
    &\hat{{\boldsymbol{\xi}}}_T^{(k,\lambda_{l},\lambda_u)} =  \hat{{\boldsymbol{\xi}}}_T^{(k)}(\lambda_l,\lambda_u) = \argmax_{\boldsymbol{\xi}_T^{(k)}\in\Theta_n} , \Tilde{L}_T^{(k)}\left(\boldsymbol{\xi}_T^{(k)}, \lambda_l, \lambda_u ; \mathbf{x}_k\right), \\
    &{\boldsymbol{\xi}}_k^* = \argmax_{\boldsymbol{\xi}_T^{(k)}\in\Theta_n}L_k\left(\boldsymbol{\xi}_T^{(k)}\right).
\end{align*}
We have 
\begin{equation}
\label{cor_3_proof_1}
    \sup _{\underline{\lambda}_l,\overline{\lambda}_u}\left|\frac{1}{T} L_T^{(k)}\left(\hat{{\boldsymbol{\xi}}}_T^{(k,\lambda_l,\lambda_u)}, \lambda_l, \lambda_u ; \mathbf{x}_k\right)-\left(\lambda_l-\lambda_u\right) L_k\left({\boldsymbol{\xi}}_k^*\right)\right| \stackrel{a.s}\rightarrow 0
\end{equation}
and
\begin{equation}
\label{cor_3_proof_2}
     \sup _{\underline{\lambda}_l,\overline{\lambda}_u}\left|{\boldsymbol{\hat{\xi}}}_T^{(k)}(\lambda_l,\lambda_u) - {\boldsymbol{\xi}}_k^0\right| \stackrel{a.s}\rightarrow 0.
\end{equation}
\end{lemma}

\begin{proof}
By the definition of ${\boldsymbol{\hat{\xi}}}_T^{(k,\lambda_l,\lambda_u)}$, we have $$\Tilde{L}_T^{(k)}\left({\boldsymbol{\hat{\xi}}}_T^{(k,\lambda_l,\lambda_u)}, \lambda_l, \lambda_u ; \mathbf{x}_k\right) \geq \Tilde{L}_T^{(k)}\left({\boldsymbol{\xi}}_k^* , \lambda_l, \lambda_u ; \mathbf{x}_k\right)$$ for all $\lambda_l,\lambda_u,T$. Combined with Lemma \ref{lemma_1} and Lemma \ref{lemma_3}, we have 
\begin{equation*}
\begin{aligned}
\left(\lambda_l-\lambda_u\right) &\left\{L_k\left({\boldsymbol{\xi}}_k^*\right) -  L_k\left(\hat{{\boldsymbol{\xi}}}_T^{(k,\lambda_l,\lambda_u)}\right)\right\} \\
\leq &\sup _{\underline{\lambda_d}, \bar{\lambda}_u}\left\{\left(\lambda_u-\lambda_l\right) L_k\left({\boldsymbol{\xi}}_k^*\right)-\frac{1}{T} \tilde{L}_T^{(k,\lambda_l,\lambda_u)}\left({\boldsymbol{\xi}}_k^*, \lambda_l, \lambda_u ; \mathbf{x}_{\mathbf{k}}\right)\right. \\
& \left.+ \frac{1}{T}\Tilde{L}_T^{(k,\lambda_l,\lambda_u)}\left(\hat{{\boldsymbol{\xi}}}_T^{(k,\lambda_l,\lambda_u)}, \lambda_l, \lambda_u ; \mathbf{x}_k\right)-\left(\lambda_u-\lambda_l\right) L_k\left(\hat{{\boldsymbol{\xi}}}_T^{(k,\lambda_l,\lambda_u)}\right)\right\} \\
= & \sup _{\underline{\lambda_d}, \bar{\lambda}_u}\left\{\biggl(\lambda_u-\lambda_l\biggr) L_k\biggl({\boldsymbol{\xi}}_k^*\biggr)-\frac{1}{T} L_T^{(k,\lambda_l,\lambda_u)}\biggl({\boldsymbol{\xi}}_k^*, \lambda_l, \lambda_u ; \mathbf{x}_{\mathbf{k}}\biggr)\right. \\
& + \left. \frac{1}{T}L_T^{(k,\lambda_l,\lambda_u)}\biggl(\hat{{\boldsymbol{\xi}}}_T^{(k,\lambda_l,\lambda_u)}, \lambda_l, \lambda_u ; \mathbf{x}_k\biggr) -\biggl(\lambda_u-\lambda_l\biggr) L_k\biggl(\hat{{\boldsymbol{\xi}}}_T^{(k,\lambda_l,\lambda_u)}\biggr)\right\} + o(1) \\
\leq & \ 2 \sup _{\underline{\lambda_d}, \bar{\lambda}_u}\sup _{\boldsymbol{\xi}_T^{(k)}\in\Theta_{n}}\left| \frac{1}{T}\Tilde{L}_T^{(k)}\left(\boldsymbol{\xi}_T^{(k)}, \lambda_l, \lambda_u ; \mathbf{x}_k\right)-\left(\lambda_u-\lambda_l\right) L_k\left(\boldsymbol{\xi}_T^{(k)}\right) \right| + o(1) \stackrel{a.s.} \rightarrow 0.
\end{aligned}
\end{equation*}
And since $L_k({\boldsymbol{\xi}}_k^*)$ is the maximum value for all ${\boldsymbol{\xi}}$, it follows that 
\begin{equation*}
    \left|L_k(\hat{{\boldsymbol{\xi}}}_T^{(k,\lambda_l,\lambda_u)}) - L_k({\boldsymbol{\xi}}_k^*) \right| \stackrel{a.s.} \rightarrow 0.
\end{equation*}
Thus, using Lemma \ref{lemma_5}, (\ref{cor_3_proof_1}) is proved. Due to the identifiability of MLE for the CARE model, (\ref{cor_3_proof_2}) is also proved.
\end{proof}

\begin{lemma}
\label{lemma_7}
Let $\boldsymbol{y} = \{y_t;t=1,\cdots, T\}$ be the observations from a PS-CARE model specified by the vector $(K_0, \boldsymbol{\lambda^0},\boldsymbol{\xi^0})$. Assume the number of change points $K_0$ is known. The estimator $(\boldsymbol{\hat{\lambda}}_T, \boldsymbol{\hat{\xi}_T})$is defined by
\begin{equation*}
    \{\boldsymbol{\hat{\lambda}}_T, \boldsymbol{\hat{\xi}_T}\} = \argmin_{\boldsymbol{\xi_T}\in\Theta_n, \boldsymbol{\lambda_T}\in A_{\epsilon}^{K_0}}\frac{2}{T}MDL(K_0, \boldsymbol{\lambda_T},\boldsymbol{\xi_T}),
\end{equation*}
where $A_{\epsilon}^{K_0}$ is defined in (\ref{lambda_assump}).  Under Assumptions~\ref{assumption_1} to~\ref{assumption_3}, for sufficiently large $T$ we have
\begin{equation*}
    \quad \boldsymbol{\hat{\lambda}}_T \stackrel{P}{\longrightarrow} \boldsymbol{\lambda}^o.
\end{equation*}
\end{lemma}

\begin{proof}
Let $B$ be the probability one set in which Lemma \ref{lemma_5} and Lemma \ref{lemma_6} hold. And we will show that $\forall \omega\in B$, we have $\boldsymbol{\hat{\lambda}}_T \rightarrow \boldsymbol{\lambda}^0$ and $\boldsymbol{\hat{\xi}}_T\rightarrow\boldsymbol{\xi}_0$. We will prove this by contradiction. Here we assume $\boldsymbol{\hat{\lambda}}_T \rightarrow \boldsymbol{\lambda}^* \neq \boldsymbol{\lambda}^0$ along a subsequence $\{T_k\}$. Also, we further assume $\boldsymbol{\hat{\xi}}_T\rightarrow\boldsymbol{\xi}^*\neq\boldsymbol{\xi}_0 $. To simplify the future notation, we replace $T_k$ with $T$. For sufficiently large $T$, we have
\begin{equation*}
    \frac{2}{T}MDL(K_0, \boldsymbol{\lambda_T},\boldsymbol{\xi_T}) = c_T - \frac{1}{T}\sum_{k=1}^{K+1}L_T^{(k)}\left(\boldsymbol{\hat{\xi}}_T^{(k)}, \hat{\lambda}_{k-1}, \hat{\lambda}_{k} ; \mathbf{y}\right),
\end{equation*}
where $c_T$ is of order $O(\log(T)/T)$. Here we simplify the notation of $\boldsymbol{{\hat{\xi}}}_T^{(k)}(\hat{\lambda}_l,\hat{\lambda}_u)$ to $\boldsymbol{\hat{\xi}}_T^{(k)}$ when there are no misunderstandings.

For each estimated interval, its limiting $I_{k}^*(\lambda_{k-1}^*, \lambda_k^*),k=1,\cdots,K+1$, there are two possible cases. The first case is when $I_k^*$ is totally contained in one true interval $(\lambda_{i-1}^0, \lambda_{i}^0)$. The second case is when $I_k^*$ covers $m+2(m\geq 0)$ true intervals $(\lambda_{i-1}^0, \lambda_{i}^0),\cdots,(\lambda_{i+m}^0, \lambda_{i+m+1}^0)$. We consider the two cases individually.

\textit{Case 1:} If $\lambda_{i-1}^0\leq \lambda_{k-1}^*\leq \lambda_{k}^*\leq\lambda_i^0$, in particular, we only consider the inequality case. Then if $\lambda_{k}^*=\lambda_{i}^0$ or $\lambda_{k-1}^*=\lambda_{i-1}^0$, as $\hat{\lambda}_{k-1}\rightarrow\lambda_{i-1}^0$ and $\hat{\lambda}_k\rightarrow \lambda_i^0$, the estimated segment can only include a decreasing proportion of observations from the adkacent segments. Then $\max(\hat{\lambda}_k-\lambda_i^0,0)$ and $\max(\lambda_{i-1}^0-\hat{\lambda}_{k-1},0)$ play the role of $h_n$ and $k_n$ in Lemma \ref{lemma_5}. So we have from Lemma \ref{lemma_5} that 
\begin{equation*}
    \frac{1}{T}L_T^{(k)}\left(\boldsymbol{{\hat{\xi}}}_T^{(k)}, \hat{\lambda}_{k-1}, \hat{\lambda}_{k} ; \mathbf{y}\right) \stackrel{a.s.} \rightarrow (\lambda_k^*-\lambda_{k-1}^*)L_i\left({\boldsymbol{\xi}}_i^0 \right).
\end{equation*}

\textit{Case 2:} If $\lambda_{i-1}^o \leq \lambda_{k-1}^*<\lambda_i^o<\ldots<\lambda_{i+k}^o<\lambda_k^* \leq \lambda_{i+m+1}^o$ for some $m\geq 0$. Then, for sufficiently large $T$, the estimated stationary process is thus non-stationary so that we can partition the likelihood by the true segment change point as below:
%\begin{equation}
%\begin{aligned}
%\label{lemma_7_proof_1}
%& \frac{1}{T} L_T^{(k)}\left(\boldsymbol{{\hat{\xi}}}_T^{(k)}, \hat{\lambda}_{k-1}, \hat{\lambda}_{k} ; \mathbf{y}\right)\\
%= & \frac{1}{T} L_T^{(k)}\left(\boldsymbol{{\hat{\xi}}}_T^{(k)}, \hat{\lambda}_{k-1}, \lambda_{i}^0 ; \mathbf{y}\right)+\frac{1}{T} \sum_{l=i}^{i+m-1} L_T^{(k)}\left(\boldsymbol{{\hat{\xi}}}_T^{(k)}, \lambda_{l}^0, \lambda_{l+1}^0 ; \mathbf{y}\right) +\frac{1}{T} L_T^{(k)}\left(\boldsymbol{{\hat{\xi}}}_T^{(k)}, \lambda_{i+k}^0, \hat{\lambda}_{k} ; \mathbf{y}\right) .
%\end{aligned}
%\end{equation}
\begin{multline}
\label{lemma_7_proof_1}
\frac{1}{T} L_T^{(k)}\left(\boldsymbol{{\hat{\xi}}}_T^{(k)}, \hat{\lambda}_{k-1}, \hat{\lambda}_{k} ; \mathbf{y}\right)
= \frac{1}{T} L_T^{(k)}\left(\boldsymbol{{\hat{\xi}}}_T^{(k)}, \hat{\lambda}_{k-1}, \lambda_{i}^0 ; \mathbf{y}\right)+ 
 \\ \frac{1}{T} \sum_{l=i}^{i+m-1} L_T^{(k)}\left(\boldsymbol{{\hat{\xi}}}_T^{(k)}, \lambda_{l}^0, \lambda_{l+1}^0 ; \mathbf{y}\right) +\frac{1}{T} L_T^{(k)}\left(\boldsymbol{{\hat{\xi}}}_T^{(k)}, \lambda_{i+k}^0, \hat{\lambda}_{k} ; \mathbf{y}\right).
\end{multline}

Each of the likelihood functions in (\ref{lemma_7_proof_1}) involves observations from one of the stationary segments. From Lemma \ref{lemma_4} and $L_i\left({\boldsymbol{\xi}}_l^0 \right)\geq L_i\left({\boldsymbol{\xi}}_k^* \right)$ for all $l=i,i+1/cdots,i+m+1$, we have 
\begin{align*}
\lim _{T \rightarrow \infty} \frac{1}{T} L_T^{(k)}\left(\boldsymbol{{\hat{\xi}}}_T^{(k)},\hat{\lambda}_{k-1},  \lambda_{i}^0; \mathbf{y}\right) &\leq\left(\lambda_i^o-\lambda_{k-1}^*\right) L_i\left({\boldsymbol{\xi}}_i^0 \right), \\
\lim _{T \rightarrow \infty} \frac{1}{T} L_T^{(k)}\left({\boldsymbol{\xi}}_T^{*},\hat{\lambda}_{l}^0,  \lambda_{l+1}^0; \mathbf{y}\right) &\leq\left(\lambda_{l+1}^o-\lambda_{l}^0\right) L_{l+1}\left({\boldsymbol{\xi}}_{l+1}^0 \right), \\ 
\lim _{T \rightarrow \infty} \frac{1}{T} L_T^{(k)}\left(\hat{{\boldsymbol{\xi}}}_T^{(k)},\hat{\lambda}_{i+m}^0,  \hat{\lambda}_{k}; \mathbf{y}\right) &\leq\left(\lambda_k^*-\lambda_{i+m}^0\right) L_{i+m+1}\left({\boldsymbol{\xi}}_{i+m+1}^0 \right).
\end{align*}

Note that the strict inequalities hold for at least one of the above equations since ${\boldsymbol{\xi}}_{k}^*$ cannot be correctly specified for all the different segments. Thus we have 
\begin{equation}
\begin{aligned}
\label{lemma_7_proof_2}
    &\lim_{T\rightarrow\infty}\frac{1}{T} L_T^{(k)}\left(\hat{{\boldsymbol{\xi}}}_T^{(k)}, \hat{\lambda}_{k-1}, \hat{\lambda}_{k} ; \mathbf{y}\right) \\
    &< \left(\lambda_i^o-\lambda_{k-1}^*\right) L_i\left({\boldsymbol{\xi}}_i^0 \right) + \sum_{l=i}^{i+m-1}\left(\lambda_{l+1}^o-\lambda_{l}^0\right) L_{l}\left({\boldsymbol{\xi}}_{l}^0 \right) + \left(\lambda_k^*-\lambda_{i+m}^0\right) L_{i+m+1}\left({\boldsymbol{\xi}}_{i+m+1}^0 \right).
\end{aligned}
\end{equation}

Now, as the number of estimated segments is the same as the true number of segments and $\lambda^* \neq \lambda^0$, there is at least one segment in which case 2 applies. Thus, for $T$ large enough, we have 
\begin{equation*}
\begin{aligned}
&\frac{2}{T}MDL(K_0, \boldsymbol{\lambda_T},\boldsymbol{\xi_T}) \\
&> \frac{c_T}{T} - \sum_{i=1}^{K+1}(\lambda_i^0-\lambda_{i-1}^0)L_{i}\left({\boldsymbol{\xi}}_{i}^0 \right) = \frac{2}{T}MDL(K_0,\boldsymbol{\lambda}^0,{\boldsymbol{\xi}}_{i}^0) \geq \frac{2}{T}MDL(K_0, \boldsymbol{\lambda_T},\boldsymbol{\xi_T}),
\end{aligned}  
\end{equation*}
which is a contradiction. Hence $\hat{\boldsymbol{\lambda}}_n\neq\boldsymbol{\lambda}^0$ for all $\omega\in B$. Thus, the lemma is proved.
\end{proof}

\begin{lemma} 
\label{lemma_8}
Under the conditions of Lemma \ref{lemma_7}, if the number of change points is unknown and is estimated from the data using (\ref{MDL_estimate}), then 
\begin{enumerate}[A.]
    \item The number of change points cannot be underestimated, which means that $\hat{K}\leq K_0$ for $T$ is large enough almost surely.
    \item When $\hat{K}>K_0$, $\boldsymbol{\lambda}^0$ must be a subset of the limit points of $\boldsymbol{\hat{\lambda}}_T$, which means for any given $\omega\in B$, $\omega>0$ and $\lambda_k^0\in\boldsymbol{\lambda}^0$, there exists a $\hat{\lambda}_i\in\boldsymbol{\hat{\lambda}}_T$ such that $|\lambda_k^0-\hat{\lambda}_i|<\epsilon$ for sufficiently large $T$. 
\end{enumerate} 
\end{lemma}

\begin{proof}
Notice that in the proof of Lemma \ref{lemma_7}, the assumption of known $K_0$ is only used to guarantee that case 2 is applied at least once. No matter how many segments $\boldsymbol{\lambda}^*$ contains, contradiction (\ref{lemma_7_proof_2}) arises whenever case 2 applies. So this lemma is proved.
\end{proof}

\begin{lemma}
\label{lemma_9}
Denote $\boldsymbol{\lambda}^0 = (\lambda_1^0,\lambda_2^0,\cdots,\lambda_{K_0}^0)$ as the true change points. Then with $(\hat{K},\boldsymbol{\hat{\lambda}_T},\boldsymbol{\hat{\xi}_T})$ defined in (\ref{MDL_estimate}), for each $k=1,2,\cdots,K_0$, there exists a $\hat{\lambda}_{i_k}\in\hat{\boldsymbol{\lambda}}_T, 1\leq i_k\leq \hat{K}$ such that for any $\delta > 0$
\begin{equation*}
    \left|\lambda_k^0-\hat{\lambda}_{i_k}\right| = O_p(T^{\delta-1}).
\end{equation*}
\end{lemma}

\begin{proof}
From Lemma \ref{lemma_8} we can assume that $\hat{K}\geq K_0$ and for each $\lambda_k^0$ there exists a $\hat{\lambda}_{i_k}$ such that $\left|\lambda_k^0-\hat{\lambda}_{i_k} \right|=o(1)$ a.s., where $1<i_1<i_2<\cdots<i_m<\hat{K}$. By construction, for every $m=0,\cdots,\hat{K}-1$, we have $\left|\hat{\lambda}_{K+1}-\hat{\lambda}_K\right|>\lambda_\epsilon$, so $\hat{\lambda}_{i_k}$ is the estimated location of change-point closets to $\lambda_k^0$ for sufficiently large $T$. We only need to prove that, for any $\delta >0$, there exists a $c>0$ such that 
\begin{equation*}
    P\left(\exists l,\left|\lambda_l^o-\hat{\lambda}_{i_l}\right|>c T^{\delta-1}\right) \rightarrow 0.
\end{equation*}

Define $\boldsymbol{\Tilde{\lambda}}_T=\{\hat{\lambda}_1,\hat{\lambda}_2,\cdots,\lambda_l^0,\hat{\lambda}_{i_l+1},\cdots,\hat{\lambda}_{K_0}\}$, where $|\lambda_l^0-\hat{\lambda}_{i_l}|>cT^{\delta-1}$. 
By the definition of $(\hat{K},\boldsymbol{\hat{\lambda}_T},\boldsymbol{\hat{\xi}_T})$, it is suffice to show that 
\begin{equation*}
    P(MDL((\hat{K},\boldsymbol{\hat{\lambda}_T},\boldsymbol{\hat{\xi}_T}))<MDL((\hat{K},\boldsymbol{\Tilde{\lambda}_T},\boldsymbol{\hat{\xi}_T})), \exists l,|\lambda_l^0- \hat{\lambda}_{i_l}|>cT^{\delta-1}) \rightarrow 0.
\end{equation*}
As the number of change points is bounded, it is sufficient to show that, for each fixed $l$, we have 
\begin{equation*}
    P(MDL((\hat{K},\boldsymbol{\hat{\lambda}_T},\boldsymbol{\hat{\xi}_T}))<MDL((\hat{K},\boldsymbol{\Tilde{\lambda}_T},\boldsymbol{\hat{\xi}_T})), |\lambda_l^0- \hat{\lambda}_{i_l}|>cT^{\delta-1}) \rightarrow 0.
\end{equation*}
Given that $|\lambda_l^0-\hat{\lambda}_{i_l}|>cT^{\delta-1}$, the difference $MDL((\hat{K},\boldsymbol{\hat{\lambda}_T},\boldsymbol{\hat{\xi}_T})) -MDL((\hat{K},\boldsymbol{\Tilde{\lambda}_T},\boldsymbol{\hat{\xi}_T}))$ is either 
\begin{align*}
    \sum_{k=T_l-[T(\lambda_l^0-\hat{\lambda}_{i_l})]+1}^{T_l}\left(l_{i_l}(\boldsymbol{{\hat{\xi}}}_{i_l}; x_{k,l},\boldsymbol{z}^*_t|\boldsymbol{y_{k,l}}) - l_{i_l+1}(\boldsymbol{{\hat{\xi}}}_{i_l+1}; x_{k,l},\boldsymbol{z}^*_t|\boldsymbol{y_{k,l}}) \right) 
\end{align*}
or
\begin{align*}
    \sum_{k=1}^{[T(\hat{\lambda}_{i_l}-\lambda_l^0)]}\left(l_{i_l+1}(\boldsymbol{{\hat{\xi}}}_{i_l+1}; x_{k,l},\boldsymbol{z}^*_t|\boldsymbol{y_{k,l}}) - l_{i_l}(\boldsymbol{{\hat{\xi}}}_{i_l}; x_{k,l},\boldsymbol{z}^*_t|\boldsymbol{y_{k,l}}) \right).
\end{align*}
%\begin{align*}
%    \sum_{k=T_l-[T(\lambda_l^0-\hat{\lambda}_{i_l})]+1}^{T_l}\left(l_{i_l}(\boldsymbol{{\hat{\xi}}}_{i_l}; x_{k,l},\boldsymbol{z}^*_t|\boldsymbol{y_{k,l}}) - l_{i_l+1}(\boldsymbol{{\hat{\xi}}}_{i_l+1}; x_{k,l},\boldsymbol{z}^*_t|\boldsymbol{y_{k,l}}) \right) \text{ or } \\
%    \sum_{k=1}^{[T(\hat{\lambda}_{i_l}-\lambda_l^0)]}\left(l_{i_l+1}(\boldsymbol{{\hat{\xi}}}_{i_l+1}; x_{k,l},\boldsymbol{z}^*_t|\boldsymbol{y_{k,l}}) - l_{i_l}(\boldsymbol{{\hat{\xi}}}_{i_l}; x_{k,l},\boldsymbol{z}^*_t|\boldsymbol{y_{k,l}}) \right).
%\end{align*}
And from Lemma \ref{lemma_1}, we have the difference is either 
\begin{align}
\label{lemma_9_prove_1}
    \Sigma^{\prime} \left(l_{i_l}(\boldsymbol{{\hat{\xi}}}_{i_l}; x_{k,l},\boldsymbol{z}^*_t|\boldsymbol{y_{k,l}}) - l_{i_l+1}(\boldsymbol{{\hat{\xi}}}_{i_l+1}; x_{k,l},\boldsymbol{z}^*_t|\boldsymbol{y_{k,l}}) \right) + O_p(1) 
\end{align}
or
\begin{align}
\label{lemma_9_prove_2}
    \Sigma^{\prime\prime} \left(l_{i_l+1}(\boldsymbol{{\hat{\xi}}}_{i_l+1}; x_{k,l},\boldsymbol{z}^*_t|\boldsymbol{y_{k,l}}) - l_{i_l}(\boldsymbol{{\hat{\xi}}}_{i_l}; x_{k,l},\boldsymbol{z}^*_t|\boldsymbol{y_{k,l}}) \right) + O_p(1),
\end{align}
%\begin{align}
%\label{lemma_9_prove_1}
%    \Sigma^{\prime} \left(l_{i_l}(\boldsymbol{{\hat{\xi}}}_{i_l}; x_{k,l},\boldsymbol{z}^*_t|\boldsymbol{y_{k,l}}) - l_{i_l+1}(\boldsymbol{{\hat{\xi}}}_{i_l+1}; x_{k,l},\boldsymbol{z}^*_t|\boldsymbol{y_{k,l}}) \right) + O_p(1) 
%    \text{ or } \\
%    \label{lemma_9_prove_2}
%    \Sigma^{\prime\prime} \left(l_{i_l+1}(\boldsymbol{{\hat{\xi}}}_{i_l+1}; x_{k,l},\boldsymbol{z}^*_t|\boldsymbol{y_{k,l}}) - l_{i_l}(\boldsymbol{{\hat{\xi}}}_{i_l}; x_{k,l},\boldsymbol{z}^*_t|\boldsymbol{y_{k,l}}) \right) + O_p(1),
%\end{align}
where $\Sigma^{\prime}= \sum_{k=T_l-[T(\lambda_l^0-\hat{\lambda}_{i_l})]+1}^{T_l}$ and $\Sigma^{\prime\prime}=\sum_{k=1}^{[T(\hat{\lambda}_{i_l}-\lambda_l^0)]}$. By Lemma \ref{lemma_6} and the ergodic theorem, we have for case (\ref{lemma_9_prove_2}) is positive and of order no less than $O(T^\delta)$ a.s. For (\ref{lemma_9_prove_1}), the ergodic theorem as well as the stationarity in each segment together guarantees the positive of this term and is of order $O(T^\delta)$. Therefore, both the quantities (\ref{lemma_9_prove_1}) and (\ref{lemma_9_prove_2}) are positive with probability going to 1, and hence, the lemma is proved.
\end{proof}

\begin{lemma}
\label{lemma_10}
If $\{X_t\}$ is a sequence of stationary, zero-mean strongly mixing process with geometric rate, and $E(|X_1|^{r+\epsilon})<\infty$ for some $2\leq r <\infty$ and $\epsilon>0$, then 
\begin{align*}
    \frac{1}{g(T)} \sum_{t=1}^{g(T)} X_t \stackrel{a . s}{\rightarrow} \mu \quad \mbox{and} \quad \frac{1}{g(T)} \sum_{t=n-g(T)+1}^{\top} X_t \stackrel{a . s}{\longrightarrow} \mu
\end{align*}
%\begin{align*}
%    &\text{i) } \frac{1}{g(T)} \sum_{t=1}^{g(T)} X_t \stackrel{a . s}{\rightarrow} \mu, \\
%&\text{ii) } \frac{1}{g(T)} \sum_{t=n-g(T)+1}^{\top} X_t \stackrel{a . s}{\longrightarrow} \mu,
%\end{align*}
for any sequence ${g(T)}_{T\geq 1}$ for integers that satisfies $g(T)>cT^{2/r}$ for some $c>0$ when $T$ is sufficiently large. Moreover, 
\begin{align*}
\sum_{t=1}^{s(T)} X_t=O\left(T^{2 / r}\right) \quad \mbox{and} \quad
\sum_{t=T-s(T)+1}^{\top} X_t=O\left(T^{2 / r}\right)
\end{align*}
%\begin{align*}
%    &\text{ iii) } \sum_{t=1}^{s(T)} X_t=O\left(T^{2 / r}\right), \\
%&\text{ iv) } \quad \sum_{t=T-s(T)+1}^{\top} X_t=O\left(T^{2 / r}\right)
%\end{align*}
a.s., for any sequence $\{s(T)\}_{T\geq 1}$ satisfying $s(T)=O(T^{2/r})$.
\end{lemma}
\begin{proof}
This lemma is taken from \cite{davis2013consistency}, from which the proof is given.
\end{proof}

\begin{lemma}
\label{lemma_11} Recall $\boldsymbol{\hat{\xi}}_T^{(k)}(\lambda_l,\lambda_u) = \argmax_{{\boldsymbol{\xi}}_k\in\Theta_n} \Tilde{L}_T^{(k)}\left({\boldsymbol{\xi}}_k, \lambda_l, \lambda_u ; \mathbf{x}_k\right)$. We have 
\begin{equation*}
    \boldsymbol{{\hat{\xi}}}_T^{(k)}(\hat{\lambda}_{k-1},\hat{\lambda}_{k}) - {\boldsymbol{\xi}}_k^0 = O\left(T^{-\frac{1}{2}}\right) \text{\quad a.s.},  
\end{equation*}
where ${\boldsymbol{\xi}}_k^0$ is the true parameter vector in $k$-th segment. 
\end{lemma}

\begin{proof}
Denote $(\hat{\lambda}_{k-1},\hat{\lambda}_k)$ by $(\lambda_l, \lambda_u)$ for simplicity. Let $\Tilde{L}_T^{'(k)}\left({\boldsymbol{\xi}}_k, \lambda_l, \lambda_u ; \mathbf{x}_k\right)$ and $\Tilde{L}_T^{''(k)}\left({\boldsymbol{\xi}}_k, \lambda_l, \lambda_u ; \mathbf{x}_k\right)$ be the first and second partial derivatives of $\Tilde{L}_T^{(k)}\left({\boldsymbol{\xi}}_k, \lambda_l, \lambda_u ; \mathbf{x}_k\right)$, respectively. Apply Taylor expansion to $\Tilde{L}_T^{'(k)}\left({\boldsymbol{\xi}}_k, \lambda_l, \lambda_u ; \mathbf{x}_k\right)$ around the true parameter vector value ${\boldsymbol{\xi}}_k^0$, we have 
\begin{equation}
\label{proof_lemma_11_1}
    \Tilde{L}_T^{'(k)}\left(\boldsymbol{{\hat{\xi}}}_k, \lambda_l, \lambda_u ; \mathbf{x}_k\right) =  \Tilde{L}_T^{'(k)}\left({\boldsymbol{\xi}}_k^0, \lambda_l, \lambda_u ; \mathbf{x}_k\right) + \Tilde{L}_T^{''(k)}\left({\boldsymbol{\xi}}_k^+, \lambda_l, \lambda_u ; \mathbf{x}_k\right)(\hat{{\boldsymbol{\xi}}}_k-{\boldsymbol{\xi}}_k^0),
\end{equation}
where $\left|{\boldsymbol{\xi}}_k^+-{\boldsymbol{\xi}}_k^0\right| < \left|\boldsymbol{{\hat{\xi}}}_k-{\boldsymbol{\xi}}_k^0\right|$. By definition of $\boldsymbol{{\hat{\xi}}}_k$, we have $L_T^{'(k)}\left(\boldsymbol{{\hat{\xi}}}_k, \lambda_l, \lambda_u ; \mathbf{x}_k\right)=0$. Therefore, (\ref{proof_lemma_11_1}) is equivalent to 
\begin{equation}
\label{proof_lemma_11_2}
    \Tilde{L}_T^{''(k)}\left({\boldsymbol{\xi}}_k^+, \lambda_l, \lambda_u ; \mathbf{x}_k\right)(\boldsymbol{{\hat{\xi}}}_k-{\boldsymbol{\xi}}_k^0) = -\Tilde{L}_T^{'(k)}\left({\boldsymbol{\xi}}_k^0, \lambda_l, \lambda_u ; \mathbf{x}_k\right).
\end{equation}
So combining Lemma \ref{lemma_1}, Lemma \ref{lemma_9} and Lemma \ref{lemma_10}, we have 
\begin{equation*}
\begin{aligned}
   \Tilde{L}_T^{'(k)}\left({\boldsymbol{\xi}}_k^0, \lambda_l, \lambda_u ; \mathbf{x}_k\right) &=  
   L_T^{'(k)}\left({\boldsymbol{\xi}}_k^0, \lambda_l, \lambda_u ; \mathbf{x}_k\right) + O(T) \\ & = \sum_{i=[T\grave{\lambda}_l]+1}^{[T\grave{\lambda}_u]}l_k^{'}\left({\boldsymbol{\xi}}_k^0; x_{i,k}|x_{l,k},l<i\right) + O_p(T^\delta) \\
   &=  \sum_{i=1}^{[T\grave{\lambda}_u]}l_k^{'}\left({\boldsymbol{\xi}}_k^0; x_{i,k}|x_{l,k},l<i\right) - \sum_{i=1}^{[T\grave{\lambda}_d]}l_k^{'}\left({\boldsymbol{\xi}}_k^0; x_{i,k}|x_{l,k},l<i\right) + O_p(T^\delta),
\end{aligned}
\end{equation*}
where $\grave{\lambda}_l=\max \left(0, \lambda_l\right)$ and $\grave{\lambda}_u=\min \left(1, \lambda_u\right)$. Since $E\left(l_k^{'}\left({\boldsymbol{\xi}}_k^0; x_{i,k}|x_{l,k}, l<i \right) \right)=0$, so the sequence $\{l_k^{'}\left({\boldsymbol{\xi}}_k^0; x_{i,k}|x_{l,k}, l<i \right)_{i\in [N]}\}$ is a stationary ergodic zero-mean martingale difference sequence with finite second moment. Thus, from \cite{pauly2011weighted} $$\sum_{i=1}^{[T\grave{\lambda}_u]}l_k^{'}\left({\boldsymbol{\xi}}_k^0; x_{i,k}|x_{l,k},l<i\right)
\quad \mbox{and} \quad 
\sum_{i=1}^{[T\grave{\lambda}_d]}l_k^{'}\left({\boldsymbol{\xi}}_k^0; x_{i,k}|x_{l,k},l<i\right)$$
are of order $O_p(T^{\frac{1}{2}})$. Thus, we have $\Tilde{L}_T^{'(k)}\left({\boldsymbol{\xi}}_k^0, \lambda_l, \lambda_u ; \mathbf{x}_k\right)=O_p(1)$ and $L_k^{''}({\boldsymbol{\xi}}_k^0)$ is positive definite. Together with Lemma \ref{lemma_5} and $\left|{\boldsymbol{\xi}}_T^+ - {\boldsymbol{\xi}}_T^0 \right|\rightarrow O_p(1)$, $\frac{1}{T}\Tilde{L}_T^{''(k)}\left({\boldsymbol{\xi}}_k^+, \lambda_l, \lambda_u ; \mathbf{x}_k\right)(\boldsymbol{{\hat{\xi}}}_k-{\boldsymbol{\xi}}_k^0)$ is positive definite. Combining all the above, the lemma is proved.
\end{proof}

\subsection{Proof of Theorem~\protect\ref{Theorem1}}
We are now ready to prove Theorem \ref{Theorem1}.
\begin{proof}
From Lemma \ref{lemma_11}, we have $\left|\boldsymbol{\hat{\xi}}_T^{(k)}(\hat{\lambda}_{k-1},\hat{\lambda}_{k}) - \boldsymbol{\xi}_k^0 \right| = O\left(T^{-\frac{1}{2}}\right) \text{ a.s.}$ Following Lemma \ref{lemma_8} and Lemma \ref{lemma_9}, it is suffice to prove that for any integer $d=1,\cdots,M-K_0$, any $\delta > 0$ and any sequence $\boldsymbol{\Tilde{\lambda}_T}=(\Tilde{\lambda}_1,\cdots,\Tilde{\lambda}_{K_0})$ such that $|\lambda_k^0-\Tilde{\lambda}_k|=O(T^{\delta-1})$ for $k=1,\cdots,K_0$,
\begin{equation}
\label{Proof_theorem_1}
    \argmin_{\boldsymbol{\xi},\boldsymbol{\lambda}\in A_{\epsilon_\lambda}^{(K_0+d)}\atop \boldsymbol{\Tilde{\lambda}_T}\subset \boldsymbol{\lambda}} \left[\frac{2}{T}MDL(K_0+d,\boldsymbol{\lambda},\boldsymbol{\xi}) \right] - \frac{2}{T}MDL(K_0,\boldsymbol{\Tilde{\lambda}_T},\boldsymbol{\xi}^0)
\end{equation}
is positive with a probability approaching 1. Denote 
$$\boldsymbol{\hat{\lambda}}_T=(\hat{\lambda}_1,\cdots,\hat{\lambda}_{K_0+d+1}) = \argmin_{\boldsymbol{\xi},\boldsymbol{\lambda}\in A_{\epsilon_\lambda}^{(K_0+d)}\atop \boldsymbol{\Tilde{\lambda}}\subset \boldsymbol{\lambda}} \left[\frac{2}{T}MDL(K_0+d,\boldsymbol{\lambda},\boldsymbol{\xi}) \right]. $$

First note that $\boldsymbol{\Tilde{\lambda}}_T\subset \boldsymbol{\hat{\lambda}}_T$ by construction. Using Taylor expansion on the likelihood function, (\ref{Proof_theorem_1}) can be rewritten as 
\begin{equation}
\begin{aligned}
\label{proof_theorem_2}
    &(\ref{Proof_theorem_1}) = C_1 - C_2 \\
    &+ \frac{1}{T}\left(\sum_{k=1}^{K_0+1}\Tilde{L}_T^{(k)}\left(\boldsymbol{\xi}_k^0, \Tilde{\lambda}_{k-1}, \Tilde{\lambda}_k ; \mathbf{y}\right) - \sum_{l=1}^{K_0+d+1}\Tilde{L}_T^{(l)}\left(\boldsymbol{\hat{\xi}}_T^{(l-1)}, \hat{\lambda}_{l-1}, \hat{\lambda}_l ; \mathbf{y}\right) \right) \\
    &- \sum_{l=1}^{K_0+d+1}\left(\boldsymbol{\hat{\xi}}_T^{(l)}-\boldsymbol{\xi}_l^* \right)^{\top}\frac{1}{T}\Tilde{L}_T^{''(k)}\left(\boldsymbol{\xi}_k^+, \hat{\lambda}_{l-1}, \hat{\lambda}_{l} ; \mathbf{y}\right)\left(\boldsymbol{\hat{\xi}}_T^{(l)}-\boldsymbol{\xi}_l^* \right),
\end{aligned}
\end{equation}
where $C_1-C_2$ is positive and of order $O(\frac{\log(T)}{T})$ and $\left|\boldsymbol{\xi}_k^+-\boldsymbol{\xi}_k^0\right| < \left|\boldsymbol{\hat{\xi}}_k-\boldsymbol{\xi}_k^0\right|$.

The third part of $(\ref{proof_theorem_2})$ is 0 and since $\left|\boldsymbol{\hat{\xi}}_T^{(l)}(\hat{\lambda}_{k-1},\hat{\lambda}_{k}) - \boldsymbol{\xi}_l^0 \right| = O\left(T^{-\frac{1}{2}}\right)$. As the fourth part is of order $O_p(T^{-1})$, $C_1-C_2$ is the dominant part in (\ref{proof_theorem_2}). Thus (\ref{Proof_theorem_1}) is positive with probability approaching 1. So, the theorem is proved.
\end{proof}

\end{appendices}

\bibliographystyle{jasa3}
\bibliography{PS-CARE}

\begin{thebibliography}{34}
\newcommand{\enquote}[1]{``#1''}
\expandafter\ifx\csname natexlab\endcsname\relax\def\natexlab#1{#1}\fi
\expandafter\ifx\csname url\endcsname\relax
  \def\url#1{{\tt #1}}\fi
\expandafter\ifx\csname urlprefix\endcsname\relax\def\urlprefix{URL }\fi

\bibitem[\protect\citeauthoryear{Aue, Cheung, Lee, and Zhong}{Aue et~al.}{2014}]{aue2014segmented}
Aue, A., Cheung, R.~C., Lee, T. C.~M., and Zhong, M. (2014), \enquote{Segmented model selection in quantile regression using the minimum description length principle,} {\em Journal of the American Statistical Association\/}, 109, 1241--1256.

\bibitem[\protect\citeauthoryear{Baltrunas, Makcinskas, and Ricci}{Baltrunas et~al.}{2010}]{baltrunas2010group}
Baltrunas, L., Makcinskas, T., and Ricci, F. (2010), \enquote{Group recommendations with rank aggregation and collaborative filtering,} in {\em Proceedings of the fourth ACM conference on Recommender Systems\/}.

\bibitem[\protect\citeauthoryear{Bong, Li, Shrotriya, and Rinaldo}{Bong et~al.}{2020}]{bong2020nonparametric}
Bong, H., Li, W., Shrotriya, S., and Rinaldo, A. (2020), \enquote{Nonparametric estimation in the dynamic {Bradley-Terry} model,} in {\em International Conference on Artificial Intelligence and Statistics\/}, PMLR.

\bibitem[\protect\citeauthoryear{Bouckaert}{Bouckaert}{1993}]{bouckaert1993probabilistic}
Bouckaert, R.~R. (1993), \enquote{Probabilistic network construction using the minimum description length principle,} in {\em Symbolic and Quantitative Approaches to Reasoning and Uncertainty: European Conference ECSQARU'93 Granada, Spain, November 8--10, 1993 Proceedings 2\/}, Springer.

\bibitem[\protect\citeauthoryear{Bradley and Terry}{Bradley and Terry}{1952}]{bradley1952rank}
Bradley, R.~A. and Terry, M.~E. (1952), \enquote{Rank analysis of incomplete block designs: {I.} The method of paired comparisons,} {\em Biometrika\/}, 39, 324--345.

\bibitem[\protect\citeauthoryear{Buyse}{Buyse}{2010}]{buyse2010generalized}
Buyse, M. (2010), \enquote{Generalized pairwise comparisons of prioritized outcomes in the two-sample problem,} {\em Statistics in Medicine\/}, 29, 3245--3257.

\bibitem[\protect\citeauthoryear{Caron and Doucet}{Caron and Doucet}{2012}]{caron2012efficient}
Caron, F. and Doucet, A. (2012), \enquote{Efficient Bayesian inference for generalized Bradley--Terry models,} {\em Journal of Computational and Graphical Statistics\/}, 21, 174--196.

\bibitem[\protect\citeauthoryear{Cattelan, Varin, and Firth}{Cattelan et~al.}{2013}]{cattelan2013dynamic}
Cattelan, M., Varin, C., and Firth, D. (2013), \enquote{Dynamic {Bradley--Terry} modelling of sports tournaments,} {\em Journal of the Royal Statistical Society: Series C (Applied Statistics)\/}, 62, 135--150.

\bibitem[\protect\citeauthoryear{Davis, Lee, and Rodriguez-Yam}{Davis et~al.}{2006}]{davis2006structural}
Davis, R.~A., Lee, T. C.~M., and Rodriguez-Yam, G.~A. (2006), \enquote{Structural break estimation for nonstationary time series models,} {\em Journal of the American Statistical Association\/}, 101, 223--239.

\bibitem[\protect\citeauthoryear{Davis and Yau}{Davis and Yau}{2013}]{davis2013consistency}
Davis, R.~A. and Yau, C.~Y. (2013), \enquote{Consistency of minimum description length model selection for piecewise stationary time series models,} {\em Electronic Journal of Statistics\/}, 7, 381--411.

\bibitem[\protect\citeauthoryear{Dwork, Kumar, Naor, and Sivakumar}{Dwork et~al.}{2001}]{dwork2001rank}
Dwork, C., Kumar, R., Naor, M., and Sivakumar, D. (2001), \enquote{Rank aggregation methods for the web,} in {\em Proceedings of the 10th International Conference on World Wide Web\/}.

\bibitem[\protect\citeauthoryear{Fan, Hou, and Yu}{Fan et~al.}{2022}]{fan2022uncertainty}
Fan, J., Hou, J., and Yu, M. (2022), \enquote{Uncertainty Quantification of {MLE} for Entity Ranking with Covariates,} {\em arXiv preprint arXiv:2212.09961\/}.

\bibitem[\protect\citeauthoryear{Glickman}{Glickman}{1999}]{glickman1999parameter}
Glickman, M.~E. (1999), \enquote{Parameter estimation in large dynamic paired comparison experiments,} {\em Journal of the Royal Statistical Society Series C: Applied Statistics\/}, 48, 377--394.

\bibitem[\protect\citeauthoryear{Glickman}{Glickman}{2001}]{glickman2001dynamic}
--- (2001), \enquote{Dynamic paired comparison models with stochastic variances,} {\em Journal of Applied Statistics\/}, 28, 673--689.

\bibitem[\protect\citeauthoryear{HAN}{HAN}{2024}]{han2024some}
HAN, Y. (2024), {\em Some contributions to uncertainty quantification and change point detection in dynamic systems\/}, Ph.D. thesis, UC Davis.

\bibitem[\protect\citeauthoryear{Han and Lee}{Han and Lee}{2024}]{10526478}
Han, Y. and Lee, T. C.~M. (2024), \enquote{Structural Break Detection in Non-Stationary Network Vector Autoregression Models,} {\em IEEE Transactions on Network Science and Engineering\/}, 11, 4134--4145.

\bibitem[\protect\citeauthoryear{Huang, Lin, and Weng}{Huang et~al.}{2006}]{huang2006ranking}
Huang, T.-K., Lin, C.-J., and Weng, R.~C. (2006), \enquote{Ranking individuals by group comparisons,} in {\em Proceedings of the 23rd International Conference on Machine Learning\/}.

\bibitem[\protect\citeauthoryear{Killick, Fearnhead, and Eckley}{Killick et~al.}{2012}]{killick2012optimal}
Killick, R., Fearnhead, P., and Eckley, I.~A. (2012), \enquote{Optimal detection of changepoints with a linear computational cost,} {\em Journal of the American Statistical Association\/}, 107, 1590--1598.

\bibitem[\protect\citeauthoryear{Lee}{Lee}{2000}]{lee2000minimum}
Lee, T. C.~M. (2000), \enquote{A minimum description length-based image segmentation procedure, and its comparison with a cross-validation-based segmentation procedure,} {\em Journal of the American Statistical Association\/}, 95, 259--270.

\bibitem[\protect\citeauthoryear{Lee}{Lee}{2001}]{lee2001introduction}
--- (2001), \enquote{An Introduction to Coding Theory and the Two-Part Minimum Description Length Principle,} {\em International Statistical Review\/}, 69, 169--183.

\bibitem[\protect\citeauthoryear{Lee}{Lee}{2002}]{leerobust02}
--- (2002), \enquote{Automatic Smoothing for Discontinuous Regression Functions,} {\em Statistica Sinica\/}, 12, 823--842.

\bibitem[\protect\citeauthoryear{Li, Rinaldo, and Wang}{Li et~al.}{2022}]{li2022detecting}
Li, W., Rinaldo, A., and Wang, D. (2022), \enquote{Detecting Abrupt Changes in Sequential Pairwise Comparison Data,} {\em Advances in Neural Information Processing Systems\/}, 35, 37851--37864.

\bibitem[\protect\citeauthoryear{Luce}{Luce}{1959}]{luce1959individual}
Luce, R.~D. (1959), \enquote{Individual Choice Behavior,} .

\bibitem[\protect\citeauthoryear{McHale and Morton}{McHale and Morton}{2011}]{mchale2011bradley}
McHale, I. and Morton, A. (2011), \enquote{A {Bradley-Terry} type model for forecasting tennis match results,} {\em International Journal of Forecasting\/}, 27, 619--630.

\bibitem[\protect\citeauthoryear{Newman}{Newman}{2023}]{newman2023efficient}
Newman, M. E.~J. (2023), \enquote{Efficient computation of rankings from pairwise comparisons,} {\em Journal of Machine Learning Research\/}, 24, 1--25.

\bibitem[\protect\citeauthoryear{Pauly}{Pauly}{2011}]{pauly2011weighted}
Pauly, M. (2011), \enquote{Weighted resampling of martingale difference arrays with applications,} .

\bibitem[\protect\citeauthoryear{Plackett}{Plackett}{1975}]{plackett1975analysis}
Plackett, R.~L. (1975), \enquote{The analysis of permutations,} {\em Journal of the Royal Statistical Society Series C: Applied Statistics\/}, 24, 193--202.

\bibitem[\protect\citeauthoryear{Rissanen}{Rissanen}{1998}]{rissanen1998stochastic}
Rissanen, J. (1998), {\em Stochastic Complexity in Statistical Inquiry\/}, World Scientific.

\bibitem[\protect\citeauthoryear{Rissanen}{Rissanen}{2007}]{rissanen2007information}
--- (2007), {\em Information and Complexity in Statistical Modeling\/}, Springer.

\bibitem[\protect\citeauthoryear{Safikhani, Bai, and Michailidis}{Safikhani et~al.}{2022}]{safikhani2022fast}
Safikhani, A., Bai, Y., and Michailidis, G. (2022), \enquote{Fast and scalable algorithm for detection of structural breaks in big var models,} {\em Journal of Computational and Graphical Statistics\/}, 31, 176--189.

\bibitem[\protect\citeauthoryear{Strobl, Wickelmaier, and Zeileis}{Strobl et~al.}{2011}]{strobl2011accounting}
Strobl, C., Wickelmaier, F., and Zeileis, A. (2011), \enquote{Accounting for individual differences in {Bradley-Terry} models by means of recursive partitioning,} {\em Journal of Educational and Behavioral Statistics\/}, 36, 135--153.

\bibitem[\protect\citeauthoryear{Wang}{Wang}{2024}]{wang2024statistical}
Wang, X. (2024), {\em Statistical Innovations in Health and Data Security: Lung Cancer Diagnosis, Microbiome Community Detection, and Adversarial Attack Analysis\/}, Ph.D. thesis, UC Davis.

\bibitem[\protect\citeauthoryear{Wang, Sharpnack, and Lee}{Wang et~al.}{2024}]{wang2024improvinglungcancerdiagnosis}
Wang, X., Sharpnack, J., and Lee, T. C.~M. (2024), \enquote{Improving Lung Cancer Diagnosis and Survival Prediction with Deep Learning and CT Imaging,} \urlprefix\url{https://arxiv.org/abs/2408.09367}.

\bibitem[\protect\citeauthoryear{Zhao and Yau}{Zhao and Yau}{2021}]{zhao2021alternating}
Zhao, Z. and Yau, C.~Y. (2021), \enquote{Alternating pruned dynamic programming for multiple epidemic change-point estimation,} {\em Journal of Computational and Graphical Statistics\/}, 30, 808--821.

\end{thebibliography}

\end{document}